%% file: main.tex
\documentclass{lmcs} %
\pdfoutput=1
\usepackage[utf8]{inputenc}

\usepackage{lastpage}
\lmcsdoi{21}{4}{29}
\lmcsheading{}{\pageref{LastPage}}{}{}%
{Aug.~16,~2024}{Dec.~12,~2025}{}


\input{macros}

\begin{document}

\title[Trade-offs between classical and quantum space]{Trade-offs between classical and quantum space\texorpdfstring{\\}{} using spooky pebbling\rsuper*}
\titlecomment{{\lsuper*}This paper is an extension of conference proceedings~\cite{quist2023optimizing}. An overview of contributions with respect to the previous version are in the `Contents and contributions' paragraph in the introduction. One should be aware of the change of terminology from ``spooks" (previous version) to ``ghosts" (this work).}

\author[A.~Quist]{Arend-Jan Quist\lmcsorcid{0000-0002-6501-2112}}
\author[A.W.~Laarman]{Alfons Laarman\lmcsorcid{0000-0002-2433-4174}}

\address{Leiden Institute of Advanced Computer Science, Leiden University, The Netherlands}	%
\email{\{a.quist, a.w.laarman\}@liacs.leidenuniv.nl}  %

\begin{abstract}
  \noindent Pebble games are used to study space/time trade-offs. Recently, spooky pebble games were introduced to study classical space / quantum space / time trade-offs for simulation of classical circuits on quantum computers. In this paper, the spooky pebble game framework is applied for the first time to general circuits. Using this framework we prove an upper bound for quantum space in the spooky pebble game. We also prove that solving the spooky pebble game is \PSPACE-complete. Moreover, we present a solver for the spooky pebble game based on satisfiability solvers combined with heuristic optimizers. This spooky pebble game solver was empirically evaluated by calculating optimal classical space / quantum space / time trade-offs. Within limited runtime, the solver could find a strategy reducing quantum space when classical space is taken into account, showing that the spooky pebble model is useful to reduce quantum space. 
\end{abstract}

\maketitle

\section{Introduction}

For a long time, scientists have been thinking how specific problems could be calculated within constraints on computational resources, such as size of the memory and runtime of the calculation. %
To study the memory usage and runtime for a calculation, pebble games were introduced. Pebble games model the use of space and time in the run of a given circuit. Trade-offs between space and time can also be easily studied by a pebble game by studying the maximum number of pebbles and time used. Such trade-offs are a fundamental problem in Computer Science~\cite{aldaz2000time,bennett1989time,swamy1978requirements,valiant1976space}\cite[Chapter 10]{savage2008models}.

\paragraph{Pebble games} In a given circuit, a pebble game is played on its underlying directed acyclic graph (DAG) structure. Each gate within the circuit corresponds to a node in this graph, with the direct input/output relationships between gates represented by directed edges. %
Every node in the graph can be pebbled or unpebbled. Pebbling a node models calculating and storing the output of the corresponding gate in memory, while unpebbling models freeing memory. %
	The primary objective of the pebble game is to pebble the entire circuit, achieved by pebbling the output gates.
Analyzing trade-offs between memory usage and runtime entails considering all strategies for pebbling the circuit. The space used by a strategy corresponds to the maximum number of pebbles employed at any given time, while the runtime is determined by the number of (un)pebble moves making up the strategy.%

There exist variations of the pebble game to model different types of computing, see for example \cite{chan2013just,nordstrom2015new} for a (non-complete) overview. In this paper, we consider three types of pebble games: irreversible pebble games, reversible pebble games and spooky pebble games. For the \textit{irreversible pebble game} all inputs of a node must be pebbled to pebble that node, but there are no constraints for unpebbling a node. This is a model for classical computing. For the \textit{reversible pebble game}, introduced by Bennett \cite{bennett1989time}, all inputs of a node must be pebbled to either pebble or unpebble that node. This is a model for reversible simulation of irreversible circuits in general and, more specific, for simulation of irreversible circuits on a quantum computer~\cite{li1998reversible}.%

\paragraph{Spooky pebble game} The \textit{spooky pebble game} is an extension of the reversible pebble game for quantum computing which was introduced recently by Gidney \cite{Gidney2019spooky} and more extensively worked out by Kornerup et al \cite{kornerup2021spooky,kornerup2024tight}. This game weakens the constraint for unpebbling a node in reversible pebbling at the cost of doing a quantum measurement. In the spooky pebble game, this is represented by a ghost, which replaces a pebble and classically stores the measurement outcome. The (classical) value of the measurement outcome is used later by the quantum computer to restore the state. Thus for the spooky pebble game we can study the trade-off between classical memory, quantum memory and time. As quantumly accessible classical space (represented by ghosts) is assumed to be much cheaper than quantum space (represented by pebbles) (see e.g. \cite{babbush2018encoding,park2019circuit,peikert2020he}), studying such a trade-off could be very useful for memory reduction in quantum computing. %

The spooky pebble game, similar to reversible pebbling, models running an irreversible (classical) computation on a quantum computer with as input a superposition. Therefore, it is a useful model for a quantum oracle calculation of a classical function. %
In many famous quantum algorithms like Shor's \cite{shor1994algorithms} and Grover's \cite{grover1996fast} such oracle calculations are essential. Thus, studying the spooky pebble game can be useful for quantum memory management for near term quantum computing on machines with limited space and time.

\paragraph{Contents and contributions} This paper is an extended version of~\cite{quist2023optimizing}. Its main contributions with respect to the previous version are 
\begin{itemize}
    \item a \PSPACE-completeness proof for solving the spooky pebbling game,\footnote{A similar result was also proven by~\cite{kornerup2024tight} after the first preprint of our work was published.}
    \item optimization algorithms for simplifying spooky pebble game instances, and
    \item a solver implementation that combines SAT solving with heuristic optimization.
\end{itemize}
In the previous version of this paper, the spooky pebble game solver consisted of a SAT solver only. The strength of SAT solvers is being very fast at finding \textit{a} solution for the game, which is not necessarily an \textit{optimal} solution. The new spooky pebble game solver is an iterative combination of SAT solving and optimization of their solutions by heuristics in a way that these complementary approaches mutually reinforce each other. The extensive experimental results show improvements to the SAT-only solvers by up to 39~percent. 

The remainder of this paper is organized as follows. In Section~\ref{sect:theory}, we explain the theory about the spooky pebble game. Note that we are the first that study the spooky pebble game on general DAGs instead of linegraphs as in \cite{kornerup2021spooky}\footnote{The results on general DAGs by~\cite{kornerup2024tight} were published after the first preprint of our work was published.}. %
In Section~\ref{sect:theorem}, we prove a theorem about the maximum memory cost for a quantum oracle computation with respect to the equivalent classical computation when there is also classical memory available. Moreover, we show that solving the spooky pebble game is \PSPACE-complete. Section~\ref{sect:pebblegameSolver} describes a solver of the spooky pebble game developed by us. 
To our knowledge this is the first solver for the spooky pebble game. In Section~\ref{sect:experiment}, the details of our open source solver implementation are presented. This section also describes experiments to create optimal quantum space, classical space and time trade-offs with the solver. %

\FloatBarrier
\section{Background: spooky pebble game}
\label{sect:theory}

In this section, we first explain measurement-based uncomputation for quantum computing, which is the basis for the spooky pebble game. Then we formally introduce (spooky) pebble games. Finally, we connect measurement-based uncomputation and spooky pebble games. Most of this section was also introduced by \cite{Gidney2019spooky,kornerup2021spooky,kornerup2024tight}. 

Any reader that is only interested in pebble games but not familiar with quantum computing can safely skip Sections~\ref{subsec:mmbasec-uncomputation}~and~\ref{subsec:connection___spookyPebbling--mmbasedUncomputation}. The main message that such skipping reader should be aware of is that in a spooky pebble game every pebble models a qubit, every ghost models a classical bit, and every move models a (semi-)quantum operation. As qubits are expensive with respect to classical bits, we usually want to reduce the number of pebbles in a game. 

\subsection{Measurement-based uncomputation}
\label{subsec:mmbasec-uncomputation}

Assume we want to simulate a classical algorithm on a superposition of classical inputs on a quantum computer. In quantum computing, the only operations one can apply are reversible operations and (possibly irreversible) measurements. Therefore, to apply irreversible (classical) operations additional qubits (ancillae) are needed. A quantum algorithm can produce a lot of garbage qubits this way. These ancillae qubits cannot be trivially removed as measuring them can destroy the superposition. Therefore, reversible uncomputation is needed to clean the ancillae.

Measurement-based uncomputation is a technique to do an irreversible uncomputing operation using a measurement such that, under certain conditions, the uncollapsed state can be restored using the outcome of the measurement. This technique can be applied for quantum oracles where the quantum input remains in memory during the entire computation and no interference gates are applied in between the measurement-based uncomputation and the restoring of the state. These conditions are satisfied when we quantumly simulate a classical oracle computation, using only ((multi-)controlled) NOT gates, which we focus on in this paper.

Figure~\ref{fig:measurement_based_uncomputation_explanation} depicts the idea of measurement-based uncomputation in a circuit. Assume that an irreversible function $f$ is computed on some quantum superposition input $\sum_x \alpha_x\ket{x}$. This is done by the reversible operation $U_f$ which writes the outcome in a new register, resulting in the state $\sum_x \alpha_x\ket{x}\ket{f(x)}$. Now, (a part of) the input is changed by an unknown but non-interfering computation, while keeping the function output in memory. If we want to remove the function output from the memory, we cannot simply uncompute the function output by applying $U_f$ again as the input has gone. Instead the function output is measured after applying a Hadamard gate to it. On measurement outcome~$b$, a phase $(-1)^{bf(x)}$ is added to every basis state of the superposition. After computing the semi-classical CNOT dependent on the measurement outcome, the memory of the function output is free and can be reused used as a $\ket{0}$ ancilla by other computations. If we want to correct for the phase $(-1)^{bf(x)}$, we need the input $\sum_x \alpha_x\ket{x}$ again and compute the function $f$. Now we can apply a $Z$-gate on the function output register, dependent on the measurement outcome $b$. Now, the function output can be uncomputed by applying $U_f$ again, and we are back in the state with input $\sum_x \alpha_x\ket{x}$.

Note that this trick of temporary uncomputing the function output of an irreversible function $f$ when the inputs are not available is not possible in usual reversible computing. Hence, measurement-based uncomputation can be seen as quantum semi-reversible computation.

For a more formal and extensive description of measurement-based uncomputation, we refer to \cite{kornerup2021spooky}.

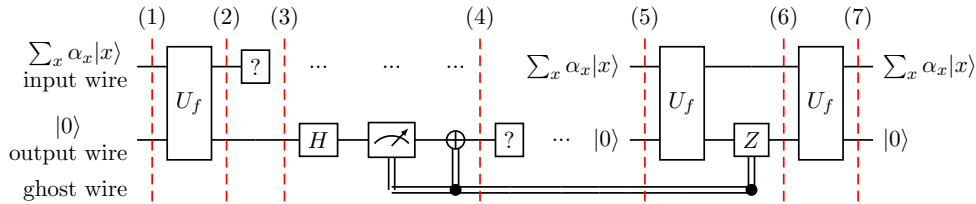
\begin{figure}[tbh]
    \centering
    \begin{tikzpicture}
    \node[scale=0.8]{
    \begin{quantikz}
    \lstick{$\sum_x \alpha_x |x\rangle$\\ input wire} \qw \slice{(1)} & \gate[wires=2]{U_f} \slice{(2)}& \gate{?} \slice{(3)} %
    & \dots & \dots & \dots \slice{(4)} & & & & & & \lstick{$\sum_x (-1)^{bf(x)} \alpha_x |x\rangle$}\slice{(5)} &  \gate[wires=2]{U_f} & \qw \slice{(6)} &  \gate[wires=2]{U_{f}} \slice{(7)}  & \qw \rstick{$\sum_x \alpha_x |x\rangle$} \\
    \lstick{$|{0}\rangle$\\ output wire} \qw &  & \qw &\gate{H} & \meter{}\vcw{1} & \targ{} & \gate{?} & & \dots & & & \lstick{$|{0}\rangle$} &  & \gate{Z} &  & \qw \rstick{$|{0}\rangle$}\\
    \lstick{ghost wire} &  & & & & \cwbend{-1} & \cw & \cw & \cw & \cw & \cw & \cw & \cw & \cwbend{-1} & &
    \end{quantikz}
    }; %
    \end{tikzpicture}
    \caption{This circuit describes measurement-based uncomputation. The three wires depict the input and output of the function $f$ and the ghost wire to classically store the measurement outcome. The labels (1), (2), etc correspond to the corresponding spooky pebbling configurations as shown in Figure~\ref{fig:spooky_pebble_configs-mmbased-uncomp}.}
    \label{fig:measurement_based_uncomputation_explanation}
\end{figure}

\subsection{Pebble games}

In this section, we formally define three types of pebble games: irreversible, reversible and spooky. 

We first define the irreversible and reversible pebble game. The irreversible pebble game is a natural model for classically computing a circuit. The reversible pebble game is a natural model for reversible computation on irreversible circuits, e.g., quantumly computing a classical circuit.

\begin{defi}[(Ir)reversible pebbling]
\label{def:irreversible_pebbling}
Let $G = (V,E)$ be a directed acyclic graph (DAG). The set of \emph{roots} $R$ is defined as $R=\{v\in V \mid \textit{out-degree}(v)=0\}$.
A \emph{(ir)reversible pebbling strategy} on $G$ is a sequence of sets of pebbled vertices $P_0,P_1,\dots,P_T\subseteq V$ such that the following conditions hold:
\begin{enumerate}
    \item $P_0 = \emptyset$;
    \item $P_T = R$;
    \item for every $t\in \{1,2,\dots,T\}$ there exists one $v\in V$ such that one of the following holds:
    \begin{itemize}
        \item \emph{pebble($v$):} $P_t=P_{t-1}\cup\{v\}$ and all direct predecessors (i.e. in-nodes) of $v$ are in $P_{t-1}$;
        \item \emph{unpebble($v$):} \\For \emph{irreversible} pebbling: $P_t=P_{t-1}\backslash\{v\}$;\\
        For \emph{reversible} pebbling: $P_t=P_{t-1}\backslash\{v\}$ and all direct predecessors (i.e. in-nodes) of $v$ are in $P_{t-1}$.
    \end{itemize}
\end{enumerate}
\end{defi}

Now, we will define the spooky pebble game, a natural model for quantumly computing a classical circuit using measurement-based uncomputation. The definition for this game is mainly adapted from \cite{kornerup2021spooky}. The spooky pebble game can be viewed as an extended generalization of the (ir)reversible pebble game. The main difference between spooky pebbling and (ir)reversible pebbling is the addition of ghosts. To regulate the ghosts, the actions \emph{ghost} and \emph{unghost} are added, and for the actions \emph{pebble} and \emph{unpebble} the requirement that no ghost is changed is added. 

\begin{defi}[Spooky pebbling]
\label{def:spooky_pebbling}
Let $G = (V,E)$ be a directed acyclic graph (DAG). The set of \emph{roots} $R$ is defined as $R=\{v\in V \mid \textit{out-degree}(v)=0\}$. 
A \emph{spooky pebbling strategy} on $G$ is a sequence of pairs of pebbles and ghost pebbles $(P_0,S_0),(P_1,S_1),\dots,(P_T,S_T)\subseteq V\times V$ such that the following conditions hold:
\begin{enumerate}
    \item $P_0 = \emptyset$ and $S_0 = \emptyset$;
    \item $P_T = R$ and $S_T = \emptyset$;
    \item for every $t\in \{1,2,\dots,T\}$ there exists one $v\in V$ such that one of the following holds:
    \begin{itemize}
        \item \emph{pebble($v$):} $P_t=P_{t-1}\cup\{v\}$ and $S_t=S_{t-1}$ and all direct predecessors (i.e. in-nodes) of $v$ are in $P_{t-1}$;%
        \item \emph{unpebble($v$):} $P_t=P_{t-1}\backslash\{v\}$ and $S_t=S_{t-1}$ and all direct predecessors (i.e. in-nodes) of $v$ are in $P_{t-1}$;%
        \item \emph{ghost($v$):} $P_t=P_{t-1}\backslash\{v\}$ and $S_t=S_{t-1}\cup\{v\}$;%
        \item \emph{unghost($v$):} $P_t=P_{t-1}\cup\{v\}$ and $S_t=S_{t-1}\backslash\{v\}$ and all direct predecessors (i.e. in-nodes) of $v$ are in $P_{t-1}$.\footnote{Note that a pebble is placed on $v$ when $v$ is unghosted.}
    \end{itemize}
\end{enumerate}
\end{defi}

We define a \emph{spooky pebbling sub-strategy} as a spooky pebbling strategy that does not necessarily obey conditions (1) and (2). We will use spooky pebbling sub-strategies as subroutines of a spooky pebbling strategy.

For a pebbling (sub-)strategy on a DAG $G$, the following three quantities are useful. The \emph{pebbling time} is the integer $T$. The \emph{pebbling cost} is $\max_{t\in[T]}|P_t|$ and (for the spooky pebbling game) the \emph{ghosting cost} is $\max_{t\in[T]}|S_t|$.

In Figure~\ref{fig:examples_of_pebbling_steps}, we see some examples of moves for the different types of pebble games. For the configuration on the top, the move \emph{unpebble(E)} can be applied for all irreversible, reversible and spooky pebbling. The move \emph{unpebble(D)} can only be applied in the irreversible pebble game, as its input vertex A is not pebbled. The move \emph{ghost(D)} can only be applied in the spooky pebble game, as the (ir)reversible pebble game does not have ghosts.
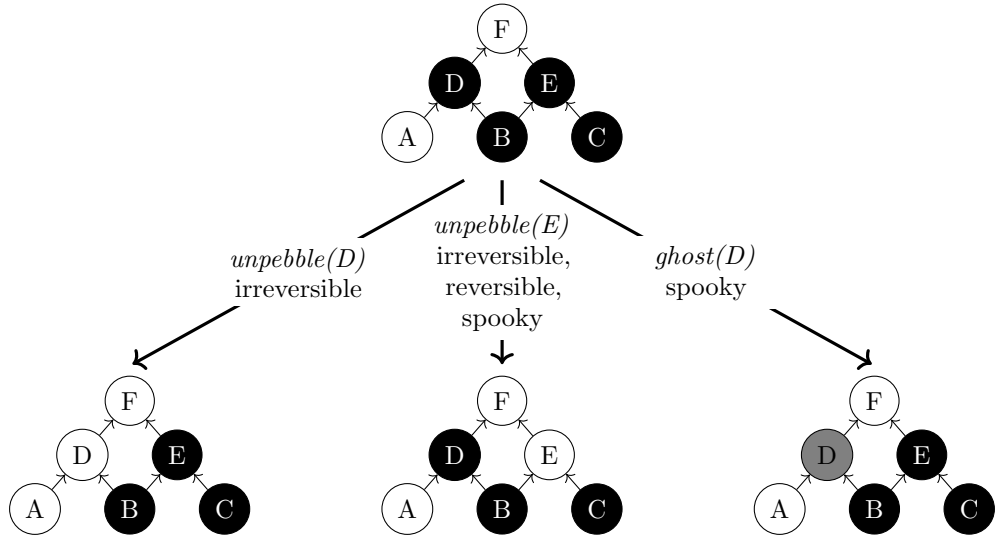
\begin{figure}[tbh]
    \centering
    \begin{tikzpicture}[scale=.9, every node/.style={scale=.9}]
    \input{Pebble_moves.tex}
    \end{tikzpicture}
    \caption{Three examples of steps for different types of pebble games. A white node represents an empty vertex, a black vertex is a pebbled vertex and a grey vertex is a ghosted vertex.}
    \label{fig:examples_of_pebbling_steps}
\end{figure}

\subsection{Connection between spooky pebble game and measurement-based uncomputation}
\label{subsec:connection___spookyPebbling--mmbasedUncomputation}

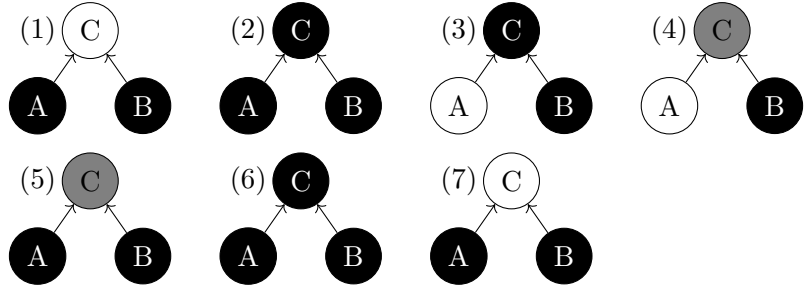
\begin{figure}[tbh]
    \centering
    \begin{tikzpicture}[scale=1.0, every node/.style={scale=1.0}]
    \input{Mmbased-uncomp-moves.tex}
    \end{tikzpicture}
    \caption{Spooky pebbling configurations for the labels in the circuit of Figure~\ref{fig:measurement_based_uncomputation_explanation}. In Definition~\ref{def:spooky_pebbling}, the spooky pebble game is introduced. The function inputs in vertices (A) and (B) are used to compute a function $f$ and store the output in vertex (C). In other words: $f\left((A),(B)\right)=(C)$. A white vertex is unpebbled, a black one is pebbled and a grey vertex is ghosted. In this example not the entire input is changed: vertex (B) is pebbled (i.e. hold in memory) all over the computation and only vertex (A) is unpebbled (i.e. removed from memory).}
    \label{fig:spooky_pebble_configs-mmbased-uncomp}
\end{figure}
In Figures~\ref{fig:measurement_based_uncomputation_explanation}~and~\ref{fig:spooky_pebble_configs-mmbased-uncomp}, the relation between the spooky pebble game and mea\-sure\-ment-based uncomputation is depicted. Quantum space is represented by pebbles and ghosts represent classical space. Like in the reversible pebble game, pebbling and unpebbling are just applying classical (irreversible) gates on a quantum computer, storing the gate-outputs in additional space. Placing a ghost can be interpreted as applying measurement-based uncomputation and storing the measurement outcome in memory. Unghosting can be interpreted as restoring the uncomputed state by recomputing the function.

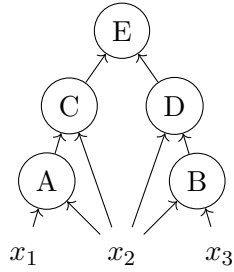
\begin{figure}[tbh]
    \centering
    \begin{tikzpicture}
        \node[circle,draw] at (-1,0) (A) {A};
        \node[circle,draw] at (1,0) (B) {B};
        \node[circle,draw] at (-.7,1) (C) {C};
        \node[circle,draw] at (.7,1) (D) {D};
        \node[circle,draw] at (0,2) (E) {E};

        \node[circle] at (-1.3,-1) (x1) {$x_1$};
        \node[circle] at (0,-1) (x2) {$x_2$};
        \node[circle] at (1.3,-1) (x3) {$x_3$};

        \draw[->] (A) -- (C);
        \draw[->] (C) -- (E);
        \draw[->] (B) -- (D);
        \draw[->] (D) -- (E);
        \draw[->] (x1) -- (A);
        \draw[->] (x2) -- (A);
        \draw[->] (x2) -- (B);
        \draw[->] (x2) -- (C);
        \draw[->] (x2) -- (D);
        \draw[->] (x3) -- (B);
    \end{tikzpicture}
    \caption{Classical circuit on input bits $x_1,x_2,x_3$. The gates A,B,C,D,E are classical gates, e.g. AND, OR.}
    \label{fig:example-irreversible-circuit}
\end{figure}

\tikzstyle{plus}=[fill=white, draw=black, shape=circle, plusss]
\tikzstyle{new style 0}=[shape=rectangle]
\tikzstyle{dtt}=[-, dotted]
\tikzstyle{red}=[-, draw=red]
\tikzstyle{red dtt}=[-, draw=red, dotted]
\tikzstyle{none}=[]
\tikzstyle{doubleline}=[-, double]
\begin{figure}[tbh]
    \centering
    \input{compiled-circuit-rev}
    \caption{Quantum circuit to compute the circuit from Figure~\ref{fig:example-irreversible-circuit} on a superposition on a quantum computer without measurement-based uncomputation.\\
    The reversible pebbling strategy that leads to this circuit consists of the following consecutive moves: pebble(A), pebble(B), pebble(C), pebble(D), unpebble(A), pebble(E), unpebble(D), unpebble(B), pebble(A), unpebble(C), unpebble(A).\\
    This strategy has a pebbling cost of 4.}
    \label{fig:compiled-circuit-reversible}
\end{figure}
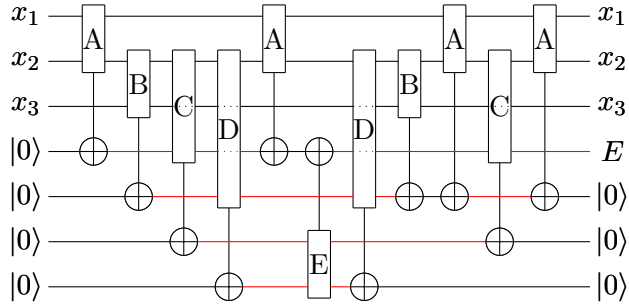

\begin{figure}[tbh]
    \centering
    \input{compiled-circuit-spooky}
    \caption{Quantum circuit to compute the circuit from Figure~\ref{fig:example-irreversible-circuit} on a superposition on a quantum computer using measurement-based uncomputation. Gate $G$ is a $H$-gate followed by a measurement in the $Z$-basis and a classically-controlled correction.\\
    The spooky pebbling strategy that leads to this circuit consists of the following consecutive moves: pebble(A), pebble(C), unpebble(A), pebble(B), pebble(D), unpebble(B), pebble(E), ghost(C), pebble(B), unpebble(D), unpebble(B), pebble(A), unghost(C), unpebble(C), unpebble(A).\\
    This strategy has a pebbling cost of 3 and a ghosting cost of 1.}
    \label{fig:compiled-circuit-spooky}
\end{figure}
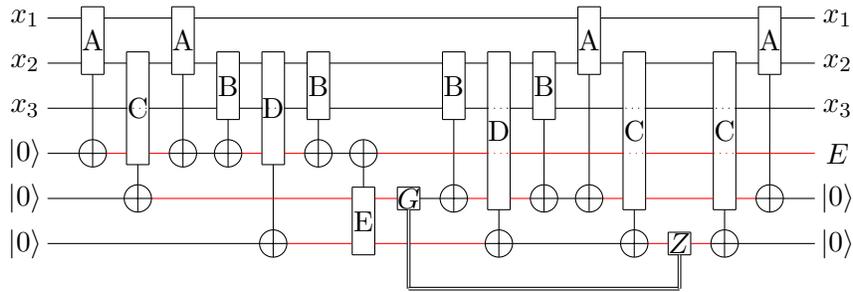

Assume we want to simulate a classical algorithm on a superposition of inputs on a quantum computer, for example the classical circuit in Figure~\ref{fig:example-irreversible-circuit}. A reversible or spooky pebbling strategy now corresponds directly to a simulation of the classical algorithm on a quantum computer. Figures~\ref{fig:compiled-circuit-reversible}~and~\ref{fig:compiled-circuit-spooky} show two such examples of quantum circuits that simulate the classical circuit of Figure~\ref{fig:example-irreversible-circuit}.

Note that for~Figure~\ref{fig:example-irreversible-circuit} the pebbling cost for the reversible pebble game is at least 4, while the pebbling cost for the spooky pebble game is 3. This shows that measurement-based uncomputation (i.e. the spooky pebble game) can reduce the number of qubits (pebbles) that is needed to simulate the circuit of~Figure~\ref{fig:example-irreversible-circuit} on a quantum superposition.

\FloatBarrier
\section{Spooky pebble cost equals irreversible pebble cost\texorpdfstring{\\}{} Spooky pebble game is \PSPACE-complete}
\label{sect:theorem}

In this section, we show that deciding whether a DAG can be spooky pebbled using $P$ pebbles is \PSPACE-complete answering the open question of~\cite{kornerup2021spooky}. For the irreversible~\cite{gilbert1979pebbling} and reversible~\cite{chan2015hardness} pebble game, this problem is already known to be \PSPACE-complete. A reduction from the irreversible pebble game is used to show \PSPACE-completeness for the spooky pebble game.

\begin{thm}
\label{thm:spooky-pebble-game=PSPACE-complete}
    The decision problem for the existence of a spooky pebbling strategy for a DAG $G$ with $P$ pebbles is \PSPACE-complete.
\end{thm}

Before we prove~Theorem~\ref{thm:spooky-pebble-game=PSPACE-complete}, we present results relating the pebbling cost of different types of pebble games presented in the previous section. This is a kind of intermezzo that is not needed to show~Theorem~\ref{thm:spooky-pebble-game=PSPACE-complete}, but the result follows directly from Lemma~\ref{lemma:irrev->spooky--cleaning_graph}~and~\ref{lemma:irrev->spooky--pebbling_root} which we use later to prove Theorem~\ref{thm:spooky-pebble-game=PSPACE-complete}. We show first that if a DAG with $m$ roots (outputs) can be irreversible pebbled with $C$ pebbles, then it can be spooky pebbled with at most $C+m$ pebbles. This result is somewhat surprising as in the irreversible strategy we can unpebble at any time. This shows the power of ghosts with respect to the reversible pebble game where a similar result does not hold.\footnote{This is easy to see, as \cite[Theorem 2]{li1998reversible} shows that reversibly pebbling a line graph needs at least logarithmic pebbles, but, trivially, irreversible pebbling of a line graph needs a constant number of pebbles.} Intuitively, this result holds because the irreversible pebble game can be simulated by the spooky pebble game with additional pebbles at all roots. Gidney~\cite{Gidney2019spooky} already proved a similar result for line graphs where an obvious optimal irreversible pebble strategy with $C=2$ pebbles can be converted to a spooky pebble strategy with $C+m=2+1=3$ pebbles, but we generalize this to arbitrary DAGs and arbitrary irreversible pebbling strategies. The statement can be formalized as follows.

\begin{thm}\label{thm:spooky-irreversible-nr.pebble-equivalence}
Let $G=(V,E)$ be a DAG with $m$ roots. Assume there exists an irreversible pebbling strategy to pebble $G$ with irreversible pebbling cost $C$ and pebbling time $T$. Then there exists a spooky pebbling strategy to pebble $G$ with spooky pebbling cost at most $C+m$ and pebbling time $T+(T+1)(|V|-m)$.  %
\end{thm}

\begin{proof}%
Let $P_0,P_1,\dots,P_T$ be an irreversible pebbling strategy with irreversible pebbling cost $C$. The result follows directly by combining Lemma~\ref{lemma:irrev->spooky--pebbling_root}~and~\ref{lemma:irrev->spooky--cleaning_graph}, using $R'=R$.
\end{proof}

\begin{lem}
\label{lemma:irrev->spooky--pebbling_root}
    Let $P_0,P_1,\dots,P_T$ be an irreversible pebbling strategy for graph $G=(V,E)$ with pebbling time $T$ and irreversible pebbling cost $C$. Then there exists a spooky pebbling sub-strategy from $(P_0,S_0)=(\emptyset,Q)$ to $(P_{T},S_{T})=(R,V\backslash R)$ with pebbling time $T$ and pebbling cost $C$, for any subset $Q\subseteq V$ .
\end{lem}
\begin{proof}
The irreversible pebbling strategy is followed with the following modifications. If a vertex is pebbled in the irreversible game, then so in the spooky game. If a vertex is unpebbled in the irreversible game, then that pebble is ghosted in the spooky game. If a vertex has a ghost when it is pebbled, the vertex is unghosted instead of pebbled. Thus at the end of this procedure, all roots (outputs) are pebbled and all other vertices are ghosted. This takes $C$ pebbles and $T$ moves, because the irreversible pebbling strategy needs $C$ pebbles and $T$ moves. 
\end{proof}

\begin{lem}
\label{lemma:irrev->spooky--cleaning_graph}
    Let $P_0,P_1,\dots,P_T$ be an irreversible pebbling strategy for graph $G=(V,E)$ with pebbling time $T$ and irreversible pebbling cost $C$. Then for every $R'\subseteq R$ there exists a spooky pebbling sub-strategy from $(P_0,S_0)=(R',V\backslash R')$ to $(P_{K},S_{K})=(R',\emptyset)$ with pebbling time $K=(T+1)(|V|-|R'|)$ and pebbling cost $C+|R'|$.
\end{lem}
\begin{proof}

    Let $v$ be a vertex of $G$ and let $G_v = (V_v,E_v) \subseteq G=(V,E)$ be the largest subgraph of $G$ rooted at $v$. We can now unpebble $v$ by projecting the irreversible strategy to $G_v$.
	For this purpose, let $P_i' = P_i \cap V_v$.
	Note, however, that all vertices in $W$ are still ghosted. 
    By Lemma~\ref{lemma:irrev->spooky--pebbling_root} for graph $G'$ with $S_0=Q=V_v$ and $R=\set{v}$, $v$ can be pebbled in pebbling time $T$ and pebbling cost $C$. Hence, unpebbling of $v$ can be done in pebbling time $T+1$ and pebbling cost $C$. 
        We remark that, after applying this spooky sub-strategy, the vertices in $V_v\backslash\{v\}$ are again ghosted and $v$ does not contain a pebble or ghost. 

    The above procedure can be done for every vertex $v$ in $V\backslash R'$. To ensure that unghosting a vertex $v$ does not interfere with other unghostings, unghosting+unpebbling is done in the reverse order as the order we pebbled them (for the first time $t$) in the irreversible pebbling strategy. 

    As we need to remove $|V|-|R'|$ pebbles, the total pebbling time is $(T+1)(|V|-|R'|)$. The pebbling cost is $C+|R'|$, as the $|R'|$ pebbles remain on the roots while $C$ pebbles are used to remove the ghosts.
\end{proof}

Now, we show that the spooky pebble game is \PSPACE-complete. To demonstrate \PSPACE-hardness, we provide a reduction from QBF, a well-known \PSPACE-complete problem~\cite{stockmeyer1973word}. 

First, we restate a theorem which shows a reduction from QBF to the irreversible pebble game for a graph $G$. This irreversible reduction is extended in Theorem~\ref{thm:spooky=QBF} by constructing a graph $G'$ related to $G$ for a reduction from QBF to the spooky pebble game.

In~\cite{gilbert1979pebbling}, for every quantified formula $Q_1x_1Q_2x_2\dots Q_nx_nF$ a graph $G$ is constructed. This graph $G$ has the properties that it is linear sized in terms of the formula length of $Q_1x_1Q_2x_2\dots Q_nx_nF$, that it has more than one edge, and that it contains a single output vertex $q_1$.

\begin{thmC}[{\cite[Theorem 1]{gilbert1979pebbling}}]
\label{thm:irrev=PSPACE}
    The quantified Boolean formula $Q_1x_1Q_2x_2\dots Q_nx_nF$ is true if and only if the graph $G$ with the properties as above has an irreversibly pebbling strategy with pebbling cost $P=3n+4$.
\end{thmC}

\begin{proof}[Proof of Theorem~\ref{thm:irrev=PSPACE}]
    In~\cite{gilbert1979pebbling} another move for the pebble game is allowed called sliding. Therefore, the proof of~\cite{gilbert1979pebbling} needs a little addition which we provide below.
    
    The sliding move that is allowed in~\cite{gilbert1979pebbling} has the following rule:\\
    ``(iii) If all predecessors of an unpebbled vertex $v$ are pebbled, a pebble may be moved from a predecessor of $v$ to $v$."\
    
    By Theorem~A from \cite{emde1978move}, for every DAG $G$ with at least one edge the pebbling cost differs exactly one between allowing and not allowing the sliding move. As $G$ in this theorem has more than one edge and the original theorem in~\cite{gilbert1979pebbling} has $P=3n+3$, Theorem~\ref{thm:irrev=PSPACE} holds for our definition of irreversible pebble games. %
\end{proof}

We will first show the existence of a DAG with desired properties that will be used in the proof of Theorem~\ref{thm:spooky=QBF}.

\begin{lem}
\label{lemma:G-diamond-P--exists}
    For every integer $P$ there exists a DAG $G^\diamond_P$ with the following properties:
    \begin{itemize}
        \item $G^\diamond_P$ has one input (leaf) vertex and one output (root) vertex
        \item $G^\diamond_P$ has size polynomial in $P$
        \item $G^\diamond_P$ can be irreversibly pebbled with $P+1$ pebbles, pebbling the input vertex only once
        \item $G^\diamond_P$ cannot be irreversibly pebbled with $P$ pebbles
        \item Every vertex in $G^\diamond_P$ except the output vertex can be irreversibly pebbled with $P$ pebbles, pebbling the input vertex only once. 
        \item $G^\diamond_P$ can be spooky pebbled with $P+1$ pebbles, and in every such spooky pebbling strategy at some point there are $P+1$ pebbles on non-input nodes and the input node must be pebbled at a later time        
    \end{itemize}
\end{lem}
\begin{proof}
    Let $G^\diamond_P$ be the following DAG:
    
    
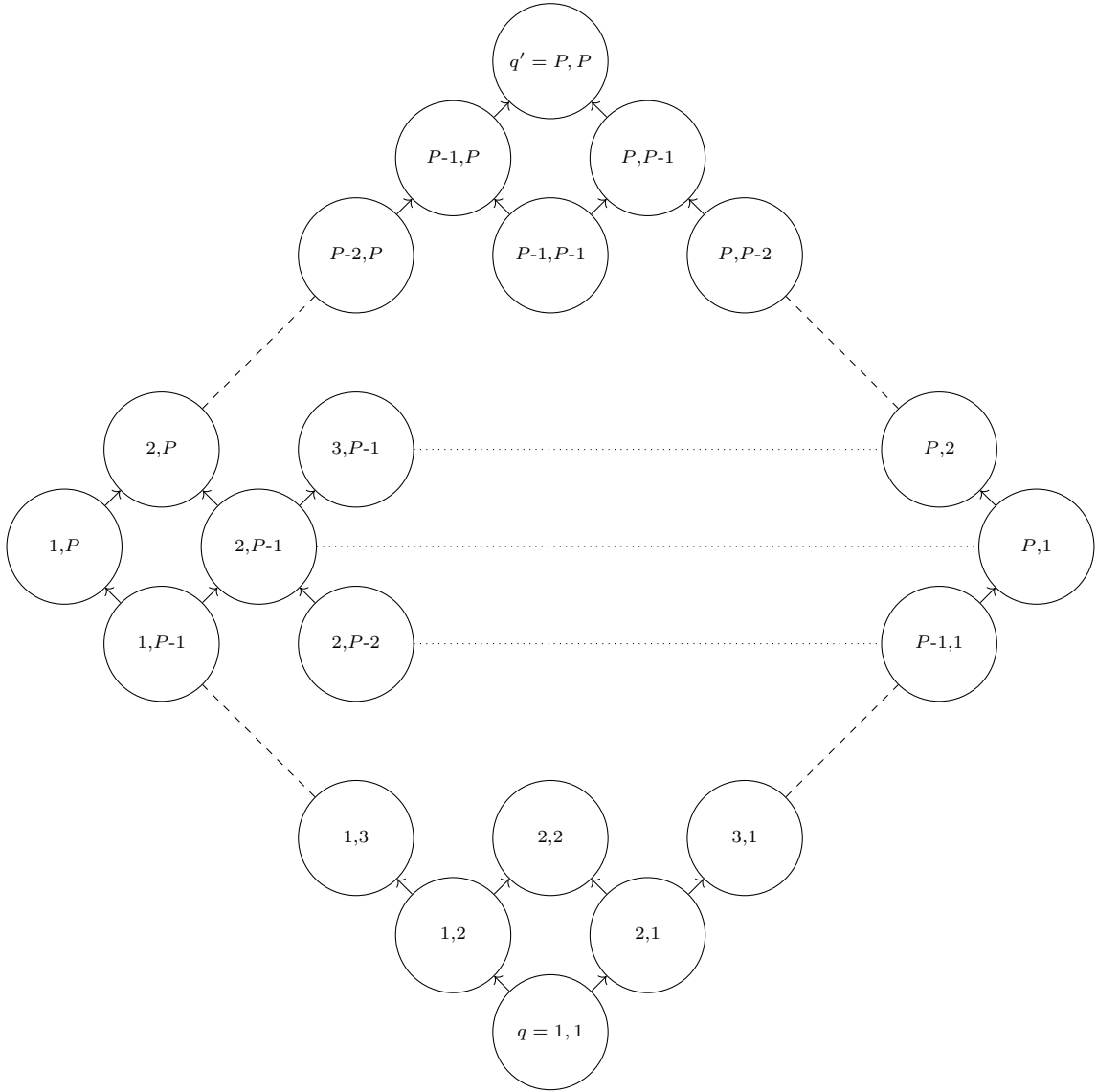
\begin{figure}[tbh]
    \centering
        \include{PSPACE-completeness/newDAGbinary_G_}
    \end{figure}

    Now we will verify that this DAG satisfies the properties from this lemma.

    It is easy to check that $G^\diamond_P$ has $P^2$ nodes, so its size is polynomial in $P$.

    Moreover, $G^\diamond_P$ has a single input and output node.

    $G^\diamond_P$ can be irreversibly pebbled with $P+1$ pebbles, pebbling the input $(1,1)$ only once, using the following strategy:

        \begin{algorithmic}[1]
            \For{$i$ in 1 to $P$}
            \State pebble($i$,1)
            \EndFor
            \For{$j$ in 2 to $P$}
            \For{$i$ in 1 to $P$}
            \State pebble($i$,$j$)
            \State unpebble($i$,$j-1$)
            \EndFor
            \EndFor
            \For{$i$ in 1 to $P-1$}
            \State unpebble($i$,$P$)
            \EndFor
        \end{algorithmic}

    Every vertex in $G^\diamond_P\backslash\{q'\}$ can be pebbled with $P$ pebbles using the above algorithm with $i$ always in 1 to $P-1$ instead of in 1 to $P$. The nodes $(P,k)$ with $k\in\{1,\dots,P-1\}$ cannot be pebbled using this algorithm, but, by symmetry of the graph, we can run the symmetric algorithm with swapped coordinates (i.e. replace pebble$(x,y)$ and unpebble$(x,y)$ with pebble$(y,x)$ and unpebble$(y,x)$) to pebble them.

    Consider an irreversible or spooky pebbling strategy to pebble $G^\diamond_P$. Consider the last time that a node on the middle line $(1,P)$,$(2,P-1)$,\dots,$(P,1)$ is pebbled before the output $q'$ is pebbled. Say that $(i,P-i+1)$ is pebbled. Then without loss of generality there exists a pebble-free path $r$ in $G^\diamond_P$ from $(i,P-i+1)$ to $q'$, otherwise, the pebbling of $(i,P-i+1)$ could have been omitted. Moreover, after $(i,P-i+1)$ is pebbled, there should be no path without pebbles from the middle line to $q'$. We have that every left diagonal path \\$(P,P)-(P-1,P)-\dots-(1,P)$; \\$(P,P-1)-(P-1,P-1)-\dots-(1,P-1)$; \\$\dots$; \\$(P,P-i)-(P-1,P-i)-\dots-(1,P-i)$\\ and every right diagonal path \\$(i+1,P)-(i+1,P-1)-\dots-(i+1,1)$; \\$\dots$; \\$(P,P)-(P,P-1)-\dots-(P,1)$ \\must be blocked by a pebble. This should happen below the path $r$, as a top part of the path $r$ can be combined with a diagonal path into a new path, which should be blocked. Therefore, at least $P-1$ pebbles are needed to block all these $P-1$ diagonals. As also 1 pebble is needed to pebble $(i,P-i+1)$, and the inputs to $(i,P-i+1)$ should be pebbled, at least $P+1$ pebbles are needed to pebble $G^\diamond_P$.\\
    Moreover, if $G^\diamond_P$ is spooky pebbled with $P+1$ pebbles, all $P+1$ pebbles are on non-input nodes. Moreover, there are no other pebbles left below the middle line, so unpebbling the input of $(i,P-i+1)$ is not possible without pebbling $q$ at a later time. When ghosting the input of $(i,P-i+1)$, this ghost should be removed before ending the game, which also requires to pebble $q$ at a later time. So the input of $G^\diamond_P$ must be pebbled at a later time to finish the game.
\end{proof}

Using the above theorem and lemma, we can prove the reduction from QBF to the spooky pebble game. This is formally stated in the following theorem. 

\begin{thm}
\label{thm:spooky=QBF}
For every Quantified Boolean Formula $F$ there exists a DAG $G'$ with one output and a number of pebbles $P$ such that $F$ is true if and only if $G'$ can be spooky pebbled with $P$ pebbles.
\end{thm}

\begin{proof}
Let $F$ be a Quantified Boolean Formula. Take from Theorem~\ref{thm:irrev=PSPACE} DAG $G$ and number of pebbles $P$ corresponding to $F$. Thus $G$ can be irreversibly pebbled with $P$ pebbles if $F$ is true, and with more than $P$ pebbles if $F$ is false. Let $G^\diamond_P$ be the DAG from Lemma~\ref{lemma:G-diamond-P--exists}.

Now, we can construct DAG $G'$ from $G$ and $G^\diamond_P$. We will show that $G'$ can be spooky pebbled with $P+1$ pebbles if $F$ is true, and with strictly more than $P+1$ pebbles if $F$ is false. Recall that $G$ has one output node, which we will name $q$. Let the input node of $G^\diamond_P$ be $q$ and let $q'$ be the output node of $G^\diamond_P$. See Figure~\ref{fig:G-prime} for the construction of $G'$.

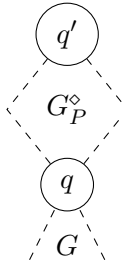
\begin{figure}[tbh]
    \centering
    \input{PSPACE-completeness/newDAG_G_}
    \caption{DAG $G'$ which can be spooky pebbled with $P+1$ pebbles if and only if $G$ can be pebbled with $P$ pebbles. See the proof in Theorem~\ref{thm:spooky=QBF}.}
    \label{fig:G-prime}
\end{figure}

We will show that $G'$ can be spooky pebbled with $P+1$ pebbles if and only if $G$ can be irreversibly pebbled with $P$ pebbles. This shows the theorem.

($\impliedby$) Assume that $G$ can be irreversibly pebbled with $P$ pebbles. So $q$ can be irreversibly pebbled with $P$ pebbles. As $G^\diamond$ can be irreversibly pebbled with $P+1$ pebbles pebbling $q$ once, we can irreversibly pebble $G'$ using $P+1$ pebbles. Similarly, every $G'\backslash\{q'\}$ can be irreversibly pebbled using $P$ pebbles. By Lemma~\ref{lemma:irrev->spooky--pebbling_root}, we can pebble $q'$ using $P+1$ pebbles, leaving ghosts at $G'\backslash\{q'\}$. As every node in $G'\backslash\{q'\}$ can be irreversibly pebbled using $P$ pebbles, every largest subgraph $G'_v$ of $G'$ with single root $v\in G'\backslash\{q'\}$ can be irreversibly pebbled with $P$ pebbles. Applying Lemma~\ref{lemma:irrev->spooky--cleaning_graph} to these graphs shows that any ghost from $G'\backslash\{q'\}$ can be removed with $P$ pebbles. \\
Hence, $G'$ can be spooky pebbled with $P+1$ pebbles.

($\implies$) Assume $G'$ can be spooky pebbled with $P+1$ pebbles. Recall that we assumed that $G^\diamond_P$ can be spooky pebbled with $P+1$ pebbles, and that in every such spooky pebbling strategy at some point there are $P+1$ pebbles on non-input nodes and the input node must be pebbled at a later time. Hence, when $G'$ is spooky pebbled, at this point $q$ must be pebbled at a later time when at least one pebble is on $G^\diamond_P\backslash\{q\}$: if all pebbles from $G^\diamond_P\backslash\{q\}$ could be removed, then the pebblings on $G^\diamond_P\backslash\{q\}$ could be omitted from the strategy. As $q$ is the output of $G$, $G$ can be spooky pebbled with $P$ pebbles.
\end{proof}

\begin{note}
\label{obs:spooky-pebble-game_in_PSPACE}
    
    It is easy to see that solving the spooky pebble game is in \PSPACE{} using the following argument (which is similar to the one in~\cite{gilbert1979pebbling} for reversible pebbling): checking whether a spooky pebbling strategy using $P$ pebbles is valid can be done in polynomial space, thus solving the spooky pebble game is in \NPSPACE. By Savitch's theorem~\cite{savitch1970relationships}, we know that $\NPSPACE=\PSPACE$. So solving the spooky pebble game is in \PSPACE.
\end{note}

\begin{proof}[Proof of Theorem~\ref{thm:spooky-pebble-game=PSPACE-complete}]
    Using the fact that QBF is \PSPACE-complete~\cite{stockmeyer1973word}, Theorem~\ref{thm:spooky=QBF} shows that the spooky pebble game is \PSPACE-hard. Note~\ref{obs:spooky-pebble-game_in_PSPACE} shows that solving the spooky pebble game is in \PSPACE. Hence, deciding whether there exists a spooky pebbling strategy for a graph $G$ with $P$ pebbles is \PSPACE-complete.
\end{proof}

\FloatBarrier
\section{Spooky pebble game solver}
\label{sect:pebblegameSolver}

As shown in the previous section, solving the spooky pebble game is a hard problem.
In this section we present a solver algorithm for the spooky pebble game that searches for strategies with a minimal cost.  To this end, in Section~\ref{sect:SAT_solver}, we encode the spooky pebble game as a satisfiability problem, similar to the method of \cite{meuli2019reversible} for the reversible pebble game.
The output of the satisfiability solver seems usually not to be an optimal pebbling strategy in number of pebbles and ghosts used or in pebbling time. Therefore, we introduce efficient and short-running heuristics in Section~\ref{sect:Optimizers}, to further optimize the pebbling strategy by reducing time, the number of pebbles and the number of ghosts.%

\subsection{SAT solver}
\label{sect:SAT_solver}

To encode the spooky pebble game into SAT, we use the Boolean variables $p_{v,i}$ and $s_{v,i}$ with $v$ the vertex and $i$ the time. The variable $p_{v,i}$ indicates whether vertex $v$ is pebbled at time $i$, and similar for ghosting with variable $s_{v,i}$. We use the shorthand $x_i = \bigcup_{v\in V} p_{v,i} \cup \bigcup_{v\in V} s_{v,i}$ to describe all variables at time $i$.

We use the clauses as stated below to describe the game. If there is an assignment of variables that satisfies all clauses, i.e. evaluates to true, this assignment describes a solution of the game. The clauses below are for the game with $T$ timesteps.\\
\textbf{Initial clauses:}
For $V$ the set of vertices in DAG $G$ we have initial clauses:
\begin{equation}
    \label{eq:init-clauses}
    I(x_0) \defn \bigwedge_{v\in V} \neg p_{v,0} \wedge \neg s_{v,0}
\end{equation}
\textbf{Final clauses:}
For $R$ the set of outputs (roots) in DAG $G$ we have final clauses:
\begin{equation}
    \label{eq:final-clauses}
    F(x_T) = \bigwedge_{v\in R} p_{v,T} \wedge \bigwedge_{v\not\in R} \neg p_{v,T} \wedge \bigwedge_{v\in V} \neg s_{v,T}
\end{equation}
\textbf{Move clauses:}
There are the following move clauses for edgeset $E$ of DAG $G$, for every time $i=0,\dots, T-1$:\\
\textit{To add a pebble }all predecessors must be pebbled and a ghost on this vertex should be removed:\footnote{Hereby we prevent a vertex to be pebbled and ghosted at the same time.}
\begin{multline}
    \label{eq:first-move-clause}
    M_1(x_i,x_{i+1})=
    \bigwedge_{(v,w)\in E} \left[\left(\neg p_{v,i} \wedge p_{v,i+1}\right)\Longrightarrow\left(p_{w,i}\wedge p_{w,i+1}\right)\right] \wedge \\
    \bigwedge_{v\in V}\left[\left(\neg p_{v,i} \wedge p_{v,i+1}\right)\Longrightarrow \neg s_{v,i+1}\right]
\end{multline}
\textit{To remove a pebble }all predecessors must be pebbled or the pebble must be changed into a ghost:
\begin{equation}
    \label{eq:second-move-clause}
    M_2(x_i,x_{i+1})=\bigwedge_{(v,w)\in E} \left[\left(p_{v,i} \wedge \neg p_{v,i+1}\right)\Longrightarrow\left((p_{w,i} \wedge p_{w,i+1}) \vee s_{v,i+1}\right)\right]
\end{equation}
\textit{To add a ghost }the vertex must be unpebbled:
\begin{equation}
    \label{eq:third-move-clause}
    M_3(x_i,x_{i+1})=\bigwedge_{v\in V} \left[\left(\neg s_{v,i} \wedge s_{v,i+1}\right)\Longrightarrow \left(p_{v,i} \wedge \neg p_{v,i+1}\right)\right]
\end{equation}
\textit{To remove a ghost }all predecessors must be pebbled and the ghost must be changed into a pebble:
\begin{equation}
    \label{eq:last-move-clause}
    M_4(x_i,x_{i+1})=\bigwedge_{(v,w)\in E} \left[\left(s_{v,i} \wedge \neg s_{v,i+1}\right)\Longrightarrow\left((p_{w,i} \wedge p_{w,i+1}) \wedge p_{v,i+1}\right)\right]
\end{equation}
\textbf{Cardinality clauses:}
For $P\in \mathds{N}$ the maximal number of pebbles and $S\in \mathds{N}$ the maximal number of ghosts we have for every $i=0,1,\dots,T$ the clauses:
\begin{equation}
    \label{eq:cardinality-clauses}
     C_{P,S}(x_i) = \left(\sum_{v\in V}p_{v,i}\leq P\right) \wedge \left(\sum_{v\in V}s_{v,i}\leq S\right)
\end{equation}

Using these clauses, we actually encode our problem as in Bounded Model Checking format (BMC)~\cite[Chapter 18]{2021Hos}. Using this formula the problem can be easily defined for arbitrary times $T$. The following BMC formula is used:
\begin{multline}
    I(x_0) \wedge C_{P,S}(x_0) \wedge M(x_0,x_1) \wedge C_{P,S}(x_1) \wedge  \dots \wedge M(x_{T-1},x_T) \wedge C_{P,S}(x_T) \wedge F(x_T).
\end{multline}
Here, $M$ is the conjunction $M = M_1\wedge M_2 \wedge M_3 \wedge M_4$ of the move clauses from Equations~\eqref{eq:first-move-clause}~to~\eqref{eq:last-move-clause}. With this BMC formula, the time $T$ can be gradually unrolled until a solution is found. This formula can be given to a SAT solver to find a solution for the spooky pebble game problem.

\begin{algorithm}[tbh]
\caption{Spooky pebble game solver using SAT (runtime timeout $t_{max}$).}
\label{alg:spooky_solver_heuristic}
\begin{algorithmic}[1]
\State $T,T_{prev} := 0,0$
\State {formula = $I(x_0)\wedge C_{P,S}(x_0)\wedge F(x_0)$}\Comment{encode pebble game for $T=0$}
\While{result $\not=$ SAT}
    \State {formula := formula.pop($F_{F_{x_{prev}}}$)}
    \State {formula := formula.push($M(x_{T_{prev}},x_{T_{prev}+1})\wedge C_{P,S}(x_{T_{prev}+1})\wedge \dots \wedge M(x_{T-1},x_{T})\wedge C_{P,S}(x_{T})\wedge F(x_T)$)}\Comment{unroll formula up to time $T$}
    \State $T_{prev} := T$
    \State result := SAT(formula) \Comment{SAT solver returns (un)sat, or timeout after $t_{wait}$ seconds}
    \If {result $=$ timeout}
        \State $T := T + T^{(skip)}$
    \EndIf
    \If {result $=$ UNSAT}
        \State $T := T + 1$
    \EndIf
\EndWhile
\end{algorithmic}
\end{algorithm}

In Algorithm~\ref{alg:spooky_solver_heuristic} we state our heuristic for the spooky pebble game solver. The solver starts with time $T=0$. If the formula is proven to be \emph{unsatisfiable} (i.e., the graph cannot be pebbled with $P$ pebbles and $S$ ghosts in $T$ moves), $T$ is incremented by 1. If the SAT solver cannot find a solution in runtime $t_{wait}$ but the formula cannot be proven to be unsatisfiable, it returns a \emph{timeout} and $T$ is incremented with $T^{(skip)}$. For a larger time $T$ the solver may find a solution faster, but this solution might not have optimal depth.\footnote{We define $T^{(skip)}$ (which is possibly greater than 1) because every call to the SAT solver takes some 'useless' time to optimize the SAT formula. Thus the pebble game solver usually can find a spooky pebble game solution faster when $T^{(skip)}>1$.} \ %
If the SAT solver finds a \emph{satisfiable} solution, the program is stopped as a strategy for the spooky pebble game is found. This entire solving procedure has a timeout of $t_{max}$.

Using push and pop statements for the final clauses, incremental SAT solving can be used, which updates the SAT formula for every new time $T$. This incremental extension of the formula speeds up the SAT solver in subsequent runs.%

We remark that the above encoding implements parallel pebbling semantics so multiple pebbles and ghosts can be added or removed at the same timestep. For example, in the pebble configuration at the top of Figure~\ref{fig:examples_of_pebbling_steps}, the steps \emph{unpebble(B)}, \emph{unpebble(C)} and \emph{pebble(F)} can be applied in one parallel timestep, see Figure~\ref{fig:parallel_semantics}. In general, the parallel pebbling time for a game is much smaller than the sequential pebbling time as multiple sequential steps can be applied in one parallel step. This highly reduces the number of BMC unrollings to find a solution. As every BMC unrolling duplicates the move and cardinality clauses as well as the variables all the time, the size of the SAT formula can be highly reduced too. This reduction of clauses and variables in the SAT formula speeds up the SAT solver. A disadvantage of this method is that the solver in general does not provide a solution with optimal sequential time.
\begin{figure}[tbh]
    \centering
    \begin{tikzpicture}[scale=1.2]
        \newcommand{\hdist}{1}    %
        \newcommand{\vdist}{0.8}    %
    
        \newcommand{\1}{0}  %
        \newcommand{\2}{1}
        \newcommand{\3}{1}
        \newcommand{\4}{1}
        \newcommand{\5}{1}
        \newcommand{\6}{0}
        \newcommand{\hshift}{0}
        \newcommand{\vshift}{0}
        \input{DAG_example.tex}

        \renewcommand{\1}{0}  %
        \renewcommand{\2}{0}
        \renewcommand{\3}{0}
        \renewcommand{\4}{1}
        \renewcommand{\5}{1}
        \renewcommand{\6}{1}
        \renewcommand{\hshift}{9}
        \renewcommand{\vshift}{2}
        \input{DAG_example.tex}

        \renewcommand{\1}{0}  %
        \renewcommand{\2}{0}
        \renewcommand{\3}{1}
        \renewcommand{\4}{1}
        \renewcommand{\5}{1}
        \renewcommand{\6}{0}
        \renewcommand{\hshift}{3}
        \renewcommand{\vshift}{-2.}
        \input{DAG_example.tex}
        
        \renewcommand{\1}{0}  %
        \renewcommand{\2}{0}
        \renewcommand{\3}{0}
        \renewcommand{\4}{1}
        \renewcommand{\5}{1}
        \renewcommand{\6}{0}
        \renewcommand{\hshift}{6}
        \renewcommand{\vshift}{-2.}
        \input{DAG_example.tex}

        \renewcommand{\1}{0}  %
        \renewcommand{\2}{0}
        \renewcommand{\3}{0}
        \renewcommand{\4}{1}
        \renewcommand{\5}{1}
        \renewcommand{\6}{1}
        \renewcommand{\hshift}{9}
        \renewcommand{\vshift}{-2}
        \input{DAG_example.tex}

        \draw [->,very thick] (1*\hdist,-.7*\vdist) -- node[text width = 2 cm, align = center, text depth = .5 cm] {\emph{unpebble(2)}} (3.2*\hdist,-.9*\vdist);
        \draw [->,very thick] (4.5*\hdist,-.9*\vdist) -- node[text width = 2 cm, align = center, text depth = .5 cm] {\emph{unpebble(3)}} (6.5*\hdist,-.9*\vdist);
        \draw [->,very thick] (7.6*\hdist,-.9*\vdist) -- node[text width = 2 cm, align = center, text depth = .5 cm] {\emph{pebble(6)}} (9.6*\hdist,-.9*\vdist);
        
        \draw [->,very thick] (1.8*\hdist,1.5*\vdist) -- node[text width = 6 cm, align = center, fill=white] {parallel move:\\ \emph{unpebble(2),unpebble(3),pebble(6)}} (9.*\hdist,3.5*\vdist);

    \end{tikzpicture}
    \caption{Example of a parallel move which consists of 3 sequential moves. Allowing parallel moves reduces the number of moves and hence reduces the size of the SAT encoding of a pebble game.}
    \label{fig:parallel_semantics}
\end{figure}
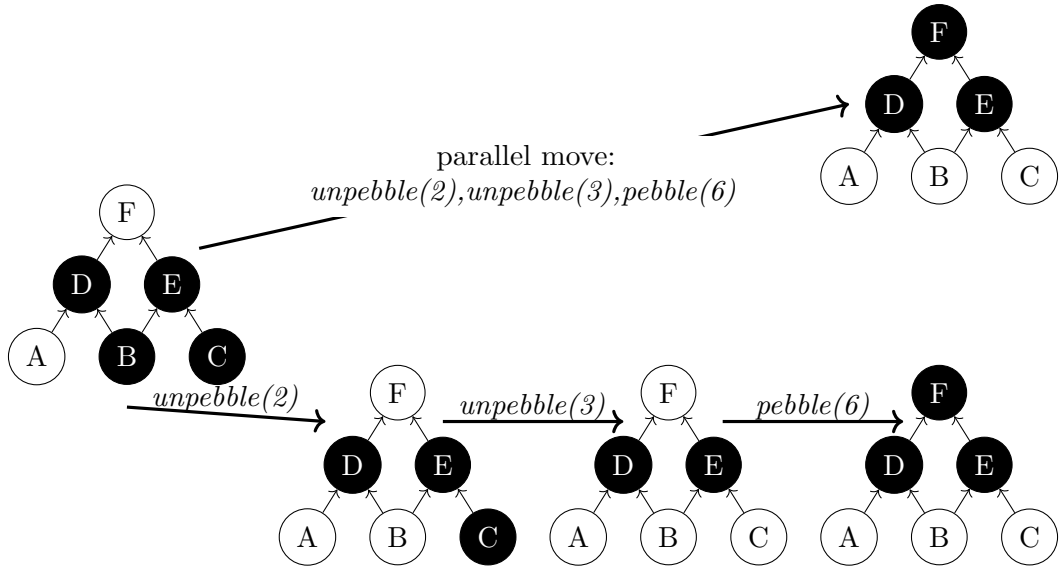

By linearizing solutions with parallel pebbling semantics, we can easily obtain a sequential solution, in which one pebble or ghost is added or removed at each time step.  These sequential times (which we usually refer to as \emph{pebbling time}, Section~\ref{sect:theory}) are usually reported for pebbling strategies, see e.g.~\cite{meuli2019reversible}, and these sequential pebbling strategies are used by our optimization heuristics in the subsequent section. %

\subsection{Heuristics for Optimization}
\label{sect:Optimizers}

The solutions to the spooky pebble game as found by SAT solving described in the previous section are usually not optimal. Often the number of pebbles, ghosts or time steps can be reduced. In this section we will describe how we optimize these solutions using a combination of heuristic methods. These heuristics reduce the number of pebbles, the number of ghosts and/or the depth of the solution.

The suboptimallity of the spooky pebble game solutions the SAT solver outputs arises from our SAT encoding, because we do not constrain the SAT solver to optimize the number of moves, the number of pebbles or the number of ghosts in a solution. Moreover, with our encoding of the spooky pebble game, the SAT solver searches only for semantically parallel solutions as this highly decreases the SAT solver runtime, but this can give inefficient sequential solutions.%

We tried to find solutions with optimal number of pebbles and ghosts by using exact methods, for example MaxSAT.\footnote{More specific, we tried the method of~\cite{narodytska2014maximum} as implemented in Z3 by~\cite{bjornerprogramming}.} These methods turned out to be very slow: they did not provide results on inputs that the SAT solver could easily deal with. This happens because exact methods only cannot find almost optimized solutions but only exactly optimized solutions, and exactly optimized solutions may be very difficult to find. Optimizing this way thus seemed to be impractical.

Therefore, we propose heuristic methods to optimize the spooky pebble game solutions  provided by the SAT solver. These heuristics take as input a spooky pebble game strategy and output an optimized strategy with less pebbles, less ghosts or less (sequential) time. Using these optimizers we can find closer Pareto fronts, such that we can find more precise trade-offs between pebbles, ghosts and (sequential) time.
We developed 6 optimization algorithms, which take all a greedy local optimum for one optimization goal such as removing all useless pebbling moves. These optimizers are described in Algorithms~\ref{alg:seqT_optimizer_heuristic}~to~\ref{alg:replace_ghost_by_pebble_asap}. Table~\ref{tab:optimizer_algos} presents an overview of all optimizers and links them to their algorithms. Applying these basic optimizers one after the other gives an improved solution with respect to the input.

\begin{table}[tbh]
    \centering
    \caption{Overview of optimization algorithms.}
    \label{tab:optimizer_algos}
    \resizebox{\columnwidth}{!}{%
    \begin{tabular}{p{0.23\textwidth}  p{0.24\textwidth}  p{0.22\textwidth}  p{0.17\textwidth} p{0.12\textwidth}}
        \toprule
        \textbf{move type} & \textbf{optimization type} & \textbf{gain} & \textbf{possible cost} & \textbf{algorithm} \\
        \midrule
        pebble & delay & pebble reduction & - & \ref{alg:delay_pebble_placement}\\
        ghost & delay & ghost reduction & pebble & \ref{alg:delay_ghost_placement}\\
        unpebble & bring forward & pebble reduction & - & \ref{alg:expedite_unpebbling}\\
        unghost & bring forward & ghost reduction & pebble & \ref{alg:replace_ghost_by_pebble_asap}\\
        pebble + unpebble & remove moves & pebble reduction & - & \ref{alg:seqT_optimizer_heuristic}\\
        ghost + unghost & remove moves & ghost reduction & pebble & \ref{alg:remove_useless_ghostings}\\
        \bottomrule
    \end{tabular}%
    }
\end{table}

These optimization heuristics all follow the same strategy: search for a type of move (such as pebble, unpebble, ghost or unghost); try to delay the move (for pebble and ghost) or to bring the move forward (for unpebble and unghost). This delay or bringing forward of moves is only possible under some conditions, e.g. for {pebbling} moves, the inputs must be pebbled at a later time and the pebble is not used to pebble or unghost another vertex.

The use of these optimizers is the following. By delaying pebbling and ghosting moves, useless pebbles or ghosts are removed. For example, if a pebble can be placed at a later time, this pebble can be used to pebble another node or be removed entirely, reducing the number of needed pebbles. Delaying ghost moves gives the same reduction for ghosts. Note that delaying a ghost move may need more pebbles, as a pebble is needed for a longer time if a ghost move is delayed. A similar reduction of pebbles resp. ghosts can be obtained by bringing unpebble and unghost moves forward. 

The optimizers that delay pebbling and bring unpebbling forward can also be combined. This results in an optimizer that removes unused pebbles by removing both the pebble and unpebble move. A similar optimizer that combines delay of ghosting and bring forward of unghosting is also defined.

The complexity of these optimizer algorithms is only linear in the time of the solution and linear in the number of nodes and edges of the graph when storing the graph by adjacency lists. %
If the number of inputs or the number of outputs of a node is unbounded, some optimizer algorithms scale linearly with these as well. We thus see that the optimizers are efficient.

\defmath\used{\mathit{used}}
\defmath\pebbled{\mathit{pebbled}}
\defmath\unpebbled{\mathit{unpebbled}}
\defmath\ghosted{\mathit{ghosted}}
\defmath\unghosted{\mathit{unghosted}}

Given a pebbling strategy $((P_0,S_0), \dots,(P_T,S_T))$,
our algorithms use the following predicates for all vertices $v\in V$ and time steps $t\in [T]$.%
\begin{align*}
\pebbled_{v,t}	&\defn v \in (P_t \setminus P_{t-1}) \cap (V \setminus S_{t-1}) &
\ghosted_{v,t}	&\defn v \in S_t \setminus S_{t-1}\\
\unpebbled_{v,t}	&\defn v \in (P_{t-1} \setminus P_{t}) \cap (V \setminus S_t) &\unghosted_{v,t}	&\defn v \in S_{t-1} \setminus S_{t}
\\
\used_{v,t}	&\defn \pebbled_{v,t}\lor  \unpebbled_{v,t} \lor \unghosted_{v,t} \hspace{-8cm}
	\\
\end{align*}

\begin{algorithm}
\caption{Remove useless pebblings from strategy $((P_0,S_0), \dots,(P_T,S_T))$.}
\label{alg:seqT_optimizer_heuristic}
\begin{algorithmic}[1]
\For {$t = T, T-1,\dots,1$}
		\For {$v \in (P_{t-1} \setminus P_{t}) \cap (V \setminus S_t)$ }\Comment{For each unpebbling move at time step $t$}
			\For {$t_0 = t-1,\dots,1$}
                \If {$\unghosted_{v,t_0} \lor \exists (w,v) \in E\colon \used_{w,t_0} $} 
					\State \textbf{break}   \Comment{$v$ was unghosted or used at time step $t_0$}
                \EndIf
				\If {$\pebbled_{v,t_0}$}
					\State $P_i$ := $P_i \setminus \set v$ \textbf{for} $i\in [t_0, t]$
					\Comment{remove pebble $v$ at time steps $t_0$ to $t$}
                \EndIf
        \EndFor
    \EndFor
\EndFor
\end{algorithmic}
\end{algorithm}

\begin{algorithm}
\caption{Remove useless ghostings from strategy $((P_0,S_0), \dots,(P_T,S_T))$.}
\label{alg:remove_useless_ghostings}
\begin{algorithmic}[1]
    \For{$t = T,T-1,\dots,1$}
        \For{$v\in S_t \setminus S_{t-1}$} \Comment{For each ghosting move at time step $t$}
            \If{$\forall (v,w)\in E: w\in (P_{t-1} \cap P_{t})$} \Comment{all inputs of $v$ are pebbled}
                \State $t_0 := t$
                \While{$v\in S_{t_0}$} \Comment{replace ghosting by unpebbling: remove ghost}
                    \State $S_{t_0} := S_{t_0}\setminus \set{v} $ 
                    \State $t_0 := t_0 + 1$
                \EndWhile
            \EndIf
        \EndFor
    \EndFor
\end{algorithmic}
\end{algorithm}

\begin{algorithm}
\caption{Delay pebble placements from strategy $((P_0,S_0), \dots,(P_T,S_T))$.}
\label{alg:delay_pebble_placement}
\begin{algorithmic}[1]
    \For{$t = T,T-1,\dots,1$}
        \For{$v\in (P_t \setminus P_{t-1}) \cap (V \setminus S_{t-1})$} \Comment{For each pebbling move at time step $t$}
                \State $t_0 := t$
                \While{$v\in P_{t_0}$ \textbf{and} $t_0<T$}
                    \If{$\exists (w,v)\in E: \used_{w,t_0} \lor \used_{w,t_0+1}$} 
                        \State \textbf{break} \Comment{$v$ was used at timestep $t_0$ or $t_0+1$}
                    \EndIf

                    \If{$\forall (v,w)\in E: w\in (P_{t_0-1}\cap P_{t_0})$}
                        \Comment{all inputs of $v$ are pebbled}
                        \State $P_{i} := P_{i} \setminus \{v\}$ \textbf{for} $i\in [t,t_0-1]$ \Comment{shift pebble move on $v$ from $t$ to $t_0$}
                    \EndIf
                    \State $t_0 := t_0 + 1$
                \EndWhile
        \EndFor
    \EndFor
\end{algorithmic}
\end{algorithm}

\begin{algorithm}
\caption{Expedite unpebblings from strategy $((P_0,S_0), \dots,(P_T,S_T))$.}
\label{alg:expedite_unpebbling}
\begin{algorithmic}[1]
    \For{$t = 1,2,\dots,T$}
        \For{$v \in (P_{t-1} \setminus P_{t}) \cap (V \setminus S_t)$} \Comment{For each unpebbling move at time step $t$}
                \State{$t_0 := t-1$}
                \While{$v \in P_{t_0}$ \textbf{and} $v\not\in S_{t_0-1}$ \textbf{and} $t_0>0$ }
                    \If{$\exists (w,v)\in E: \used_{w,t_0} \lor \used_{w,t_0+1}$} 
                        \State \textbf{break} \Comment{$v$ was used at timestep $t_0$ or $t_0+1$}
                    \EndIf

                    \If{$\forall (v,w)\in E: w\in (P_{t_0-1}\cap P_{t_0})$}
                        \Comment{all inputs of $v$ are pebbled}
                        \State $P_{i} := P_{i} \setminus \{v\}$ \textbf{for} $i\in [t_0,t-1]$ \Comment{shift unpebble move on $v$ from $t$ to $t_0$}
                    \EndIf
                    \State $t_0 := t_0 - 1$

                \EndWhile
        \EndFor
    \EndFor
\end{algorithmic}
\end{algorithm}

\begin{algorithm}
\caption{Delay ghost placements from strategy $((P_0,S_0), \dots,(P_T,S_T))$.}
\label{alg:delay_ghost_placement}
\begin{algorithmic}[1]
    \For{$t = T,T-1,\dots,1$}
        \If{$|P_t| < P$} \Comment{If not all pebbles are used at time $t$}
            \For{$v\in S_t \setminus S_{t-1}$} \Comment{For each ghosting move at time step $t$}
                \State $P_t := P_t \cup \set{v}$ \Comment{remove ghost and place pebble on $v$ at time $t$}
                \State $S_t := S_t \setminus \set{v}$
                \If{$|P_t| \geq P$} \Comment{Check current availability of pebbles}
                    \State \textbf{break}
                \EndIf
            \EndFor
        \EndIf
    \EndFor
\end{algorithmic}
\end{algorithm}

\begin{algorithm}
\caption{Replace ghost by pebble as soon as possible from strategy $((P_0,S_0), \dots,(P_T,S_T))$.}
\label{alg:replace_ghost_by_pebble_asap}
\begin{algorithmic}[1]
    \For{$t = 1,2,\dots,T$}
        \If{$|P_t|<P$}
            \For{$v \in S_{t-1} \setminus S_{t}$} \Comment{For each unghosting move at time step $t$}
                    \State $t_0:=t-1$
                    \While{$v\in S_{t_0}$ \textbf{and} $t_0>0$}
                        \If{$|P_t|\geq P$} \Comment{Check current availability of pebbles}
                            \State \textbf{break}
                        \EndIf

                        \If{$\forall (v,w)\in E: w\in (P_{t_0-1}\cap P_{t_0})$}
                            \Comment{all inputs of $v$ are pebbled}
                            \For{$i\in [t_0,t-1]$ } \Comment{shift unghost move on $v$ from $t$ to $t_0$}
                                \State $P_{i} := P_{i} \cup \{v\}$ 
                                \State $S_{i} := S_{i} \setminus \{v\}$ 
                            \EndFor
                        \EndIf
                        \State $t_0 := t_0 - 1$
                    \EndWhile
            \EndFor
        \EndIf
    \EndFor
\end{algorithmic}
\end{algorithm}

\FloatBarrier
\section{Implementation and Experiment}
\label{sect:experiment}

We implemented our spooky game solver algorithm from Section~\ref{sect:pebblegameSolver} in Python. The code is open source and can be found in \cite{quist2025sourcecode}. The Z3 solver for satisfiability is used as SAT solver \cite{de2008z3}. %
The optimizer algorithms were implemented in Python.

\begin{table}[tbh]
    \centering
    \caption{Properties of the DAGs of the ISCAS85 benchmarks used to benchmark the spooky pebble game solver. }
    \begin{tabular}{cccc}
    \toprule
    name & vertices & roots & edges\\
    \midrule
    c432 & 172 & 7 & 260\\
    c499 & 177 & 32 & 246\\
    c880 & 276 & 26 & 374\\
    c1355 & 177 & 32 & 246\\
    c1908 & 193 & 25 & 257\\
    c2670 & 401 & 46 & 533\\
    c3540 & 830 & 22 & 1395\\
    c5315 & 1089 & 96 & 1503\\
    c6288 & 979 & 32 & 1639\\
    c7552 & 988 & 62 & 1584\\
    \bottomrule
    \end{tabular}
    \label{tab:ISCAS_information}
\end{table}

To test the performance of our implementation, we use an Intel Core
i7-6700 CPU with eight cores clocked at 3.40GHz. We use the well-known ISCAS85 benchmarks circuits~\cite{brglez1985neutral}. Using the open source tool mockturtle, we transformed the circuits via XOR-majority graphs (DAGs) to circuits with common quantum gates \cite{soeken2017design}. Table~\ref{tab:ISCAS_information} shows the characteristics of the resulting DAGs.

For all these benchmarks we measure the performance of the pebble game solver. The goal is to find a Pareto front of spooky pebble game solutions with optimal parameters (sequential) time, number of pebbles and number of ghosts. The procedure is described in detail in Algorithms~\ref{alg:pareto_front_searcher_for_benchmarks_algo}~and~\ref{alg:optimize_solution}. For various numbers of ghosts a solution is searched with a decreasing number of pebbles. For each constraint, the solver is run three times with different seeds. If at least one solution is found, the number of pebbles is decreased by five. If no solution is found, the number of pebbles ($P$) is not decreased anymore and the iteration with next number ghosts ($S$) is started.

\begin{algorithm}[tbhp]
\caption{Optimal solutions search for benchmarks.
Function spooky\_pebble\_game\_solver(pebbles, ghosts) is described in Algorithm~\ref{alg:spooky_solver_heuristic}.}
\label{alg:pareto_front_searcher_for_benchmarks_algo}
\begin{algorithmic}[1]
    \For {$S$ = vertices, vertices/5, vertices/10, vertices/20, 0}
        \State spooky\_pebble\_game\_solver($\infty$, $S$) \Comment{run solver five times}
        \State pb := minimal \#pebbles in found solutions
        \For{$P$ = pb, pb-5, pb-10, \dots}
            \State spooky\_pebble\_game\_solver($P$, $S$) \Comment{run solver five times}
            \If{no solution is found by solver}
                \State \textbf{break}
            \EndIf
        \EndFor
    \EndFor
\end{algorithmic}
\end{algorithm}

As runtime parameters for the solver (Algorithm~\ref{alg:spooky_solver_heuristic}), we use $t_{wait}=15$ seconds and $t_{max}=2$ minutes. The solver gives little results for these runtimes on large benchmarks.\footnote{Large benchmarks need many clauses and variables to encode one BMC iteration, and usually also take a larger pebbling time (more BMC iterations). Hence, their SAT formula is much larger, which makes it a harder job for the SAT solver to find a solution. }

Thus, we used a longer runtime for the large benchmarks. For the four largest benchmarks, $t_{wait}=60$ seconds and $t_{max}=8$ minutes is used. For all benchmarks, $T^{(skip)}$ is set to 5.

\begin{algorithm}[tbhp]
\caption{Transform parallel solution $((P_0,S_0), \dots,(P_T,S_T))$ with $P$ pebbles and $S$ ghosts obtained from SAT solver into sequential solution $((P'_0,S'_0), \dots,(P'_{T'},S'_{T'}))$.}
\label{alg:sequentializer}
\begin{algorithmic}
    \State $(P'_0,S'_0):=(P_0,S_0)$
    \State $T':=1$
    \For{$t=0,1,\dots,T-1$}
        \State pebbling, unpebbling, ghosting, unghosting = \emph{empty list}
        \For{$v\in V$}
            \If{$v\in P_{T+1}\setminus P_{T}$} \Comment{$v$ is pebbled}
                \State pebbling.push($v$)
            \EndIf
            \If{$v\in P_T\setminus P_{T+1}$} \Comment{$v$ is unpebbled}
                \State unpebbling.push($v$)
            \EndIf
            \If{$v\in S_{T+1}\setminus S_{T}$} \Comment{$v$ is ghosted}
                \State ghosting.push($v$)
            \EndIf
            \If{$v\in S_T\setminus S_{T+1}$} \Comment{$v$ is unghosted}
                \State unghosting.push($v$)
            \EndIf
        \EndFor
        \For{$v\in$unpebbling}
            \State $(P'_{T'},S'_{T'}) := (P'_{T'-1}\setminus \{v\},S'_{T'-1})$
            \State $T':=T'+1$
        \EndFor
        
        \While{$|$ghosting$|>0$ \textbf{or} $|$unghosting$|>0$}
            \If{$|P'_{T'-1}|\geq P$ \textbf{and} $|S'_{T'-1}|\geq S$} \Comment{ghost and unghost in one timestep}
                \State $v:=$unghosting.pop()
                \State $w:=$ghosting.pop()
                \State $(P'_{T'},S'_{T'}) := ((P'_{T'-1}\cup\{v\})\setminus \{w\},(S'_{T'-1}\cup\{w\})\setminus\{v\})$
                \State $T':=T'+1$
            \EndIf
            \If{$|$unghosting$|>0$ \textbf{and} $|P'_{T'-1}|<P$}
                \State $v:=$unghosting.pop()
                \State $(P'_{T'},S'_{T'}) := (P'_{T'-1}\cup\{v\},S'_{T'-1}\setminus\{v\})$
                \State $T':=T'+1$
            \EndIf

            \If{$|$ghosting$|>0$ \textbf{and} $|S'_{T'-1}|<S$}
                \State $v:=$ghosting.pop()
                \State $(P'_{T'},S'_{T'}) := (P'_{T'-1}\setminus\{v\},S'_{T'-1}\cup\{v\})$
                \State $T':=T'+1$
            \EndIf
        \EndWhile

        \For{$v\in$pebbling}
            \State $(P'_{T'},S'_{T'}) := (P'_{T'-1}\cup \{v\},S'_{T'-1})$
            \State $T':=T'+1$
        \EndFor
    \EndFor
\end{algorithmic}
\end{algorithm}

\begin{algorithm}[tbhp]
\caption{Optimize solution $((P_0,S_0), \dots,(P_T,S_T))$ obtained from SAT solver.
Other orders of optimizers are possible, this is random.}
\label{alg:optimize_solution}
\begin{algorithmic}[1]
    \State solution := $((P_0,S_0), \dots,(P_T,S_T))$
    \While{solution improved} \Comment{time, pebbles or ghosts reduced w.r.t. previous run}
        \State solution := remove\_ghostings(solution)
        \State solution := remove\_pebblings(solution)
        \State solution := delay\_pebbling(solution)
        \State solution := bring\_forward\_unpebbling(solution)
        \State solution := delay\_ghosting(solution)
        \State solution := bring\_forward\_unghosting(solution)
    \EndWhile
\end{algorithmic}
\end{algorithm}

For every solution the SAT solver found, we sequentialize the parallel solution from the SAT solver using Algorithm~\ref{alg:sequentializer}. After sequentialization, we run the optimizer algorithm as described in Algorithm~\ref{alg:optimize_solution} on the solution. For every pebbling strategy the SAT solver found, this algorithm is run five times, each time with a different order of optimizations. If the solution is not optimized by the six optimization operations, the optimization algorithm is stopped. 

\begin{table}[tbh]
    \centering
    \caption{Maximal runtimes of pebble game (PG) solver and optimization heuristic per benchmark. }
    \label{tab:Runtime_results}
    \begin{tabular}{cccc}
        \toprule
        name &  PG solver & \hspace*{.5em}& heuristic \\
        \midrule
        c432 & 121.3 && 13.7\\
        c444 & 111.3 && 11.5\\
        c880 & 117.8 && 44.5\\
        c1355 & 110.4 && 13.8\\
        c1908 & 108.8 && 19.4\\
        c2670 & 118.5 && 97.9\\
        c3540 & 464.6 && 814.5\\
        c5315 & 397.7 && 802.8\\
        c6288 & 361.9 && 1341.9\\
        c7552 & 46.2 && 146.1\\
        \bottomrule
    \end{tabular}
\end{table}

In Figure~\ref{fig:ParetoFront_results}, the results of the runs of the spooky pebble game solver are depicted. The color of the datapoint indicates the number of ghosts used by the solution. The large dots indicate the solutions found by the SAT solver. The smaller dots indicate the pebbling strategies found by optimizing the SAT solver solutions. Notice that there are many optimized pebbling strategies as all solutions found during Algorithm~\ref{alg:optimize_solution} are depicted.

\begin{figure}[htbp]
    \centering
    \includegraphics[trim={0 .5cm 0 0},width=\textwidth,height=\textheight-6em-1.32pt]{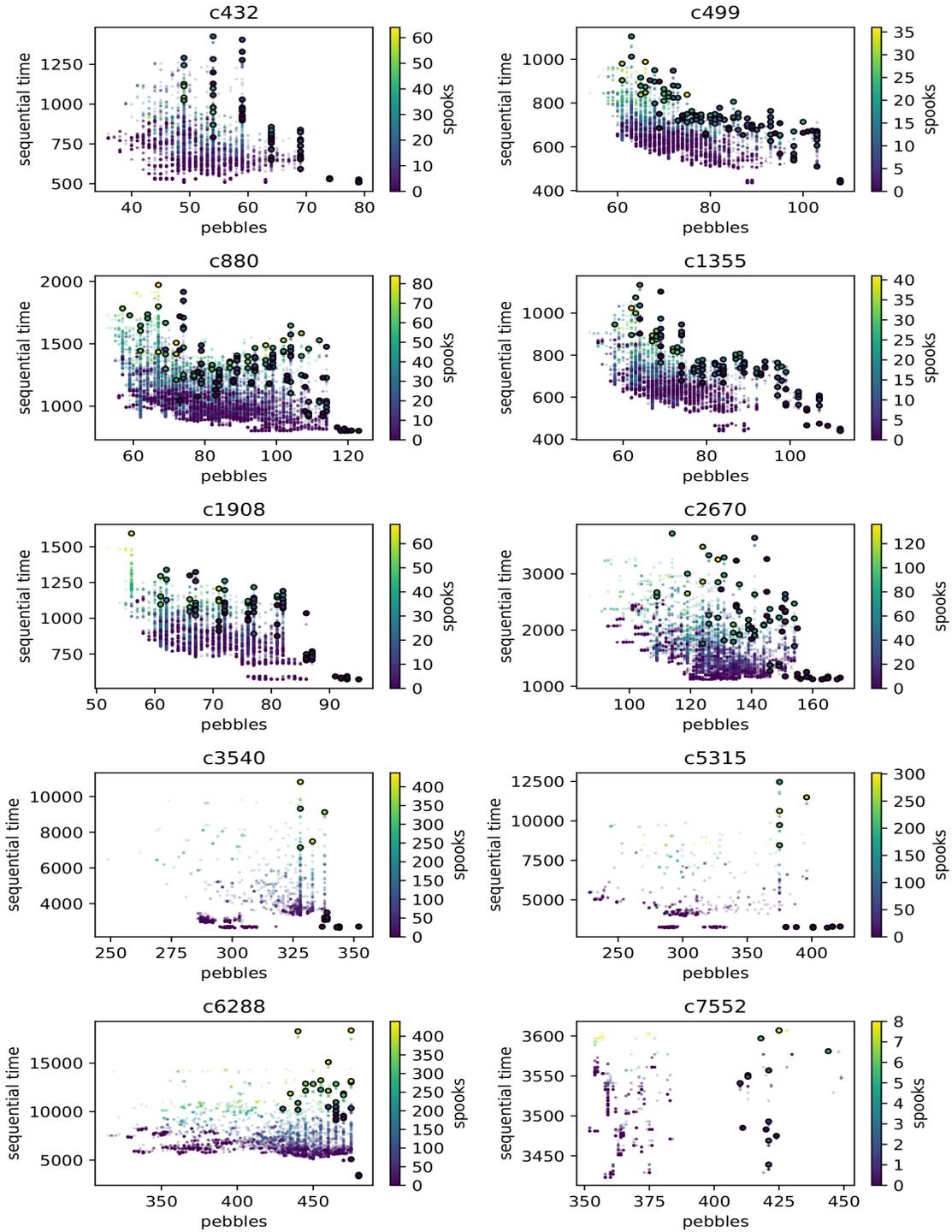}
    \caption{Solutions of the spooky pebble game with sequential time, number of pebbles and number of ghosts found by the SAT solver heuristic of Algorithm~\ref{alg:pareto_front_searcher_for_benchmarks_algo} (big dots) and their optimized solutions as found by the optimizer heuristics of Algorithm~\ref{alg:optimize_solution} (small dots).  }
    \label{fig:ParetoFront_results}
\end{figure}

Table~\ref{tab:Runtime_results} shows the maximal runtimes of the SAT solver and the optimizer to find a solution, respectively to optimize the solution. The maximal runtimes of the SAT solver are indeed bounded by 2 minutes (for the 6 smallest benchmarks) and 8 minutes (for the 4 largest benchmarks).\footnote{The runtime for c432 is marginally longer, as the time for loading the benchmark into the solver is also taken into account. This loading takes some additional seconds.} The maximal runtimes of the optimizers are strongly correlated to the maximal sequential time of the found solutions. This is because the time complexity of the optimizers goes linear in the sequential time of the solution. The optimizers can be sped up by implementing the optimizers in e.g. C or C++ instead of Python, as their execution involves extensive use of for-loops. Therefore, these long runtimes of the optimizers seem to be no bottleneck or problem for the total runtime of the solver.

The data in Figure~\ref{fig:ParetoFront_results} shows that in all cases the optimizers reduce the number of ghosts and sequential time of the pebbling strategies with respect to the strategies found by the SAT solver. Moreover, the number of pebbles is reduced by the optimizers. This reduction of pebbles is significant (up to 39\%). Thus, the SAT solver is a good starting point to find pebbling strategies, but the optimizers can improve the solution substantially.

These results show a well-shaped front of solutions without unexpected jumps or irregular discontinuities or non-convexity. This improves upon the results of the previous version of this paper~\cite{quist2023optimizing}. 

Moreover, we observe that many optimized solutions with minimal sequential time use (almost) no ghosts. This ghost reduction is much better than the SAT solver pebbling strategies, as provided in the previous version of this paper~\cite{quist2023optimizing}. Nevertheless, for most benchmarks the pebbling strategy with the least number of pebbles uses ghosts. Thus, optimization of quantum space is still obtained by using classical space.

For the cases where the pebbling strategy for a specific number of pebbles with minimal sequential time uses non-zero ghosts, the sequential time is much less than in similar cases with zero ghosts. See e.g. 62 pebbles in c880 or 67 in c1355. This indicates that not reducing the number of ghosts may reduce the sequential time. But as our optimization algorithms usually reduce the number of ghosts, this behavior is rarely observed.

\FloatBarrier
\section{Conclusions and further research}
\label{sect:conclusion}

In this paper, we explored the spooky pebble game, a model for the trade-off between quantum space, classical space and time in a quantum computation. Using the spooky pebble game, we proved that a classical calculation that uses space $C$ can be executed by a quantum computer using quantum space $C+m$ and some additional classical space, where $m$ is the space of the output. 
In addition, we showed that solving the spooky pebble game is \PSPACE-complete.

Moreover, we showed the design and implementation of a spooky pebble game solver, which gives a memory management strategy for a quantum computation with a limited amount of quantum space and classical space. This solver uses a SAT solver to initially find pebbling strategies. Using this exact SAT solver method, optimal solutions are very hard to find. The (non-optimal) strategies found by the SAT solver are combined with heuristics to optimize these moderate solutions into significantly better local optima. Thus the SAT solver and the heuristics are complementary: the SAT solver can only find moderate solutions, but the heuristics can optimize this by pebble reductions up to 39\%. This is not surprising as solving the spooky pebble game is \PSPACE-complete, and SAT solving is used here as an NP-oracle. Although solving a spooky pebble game is \PSPACE-complete, this solver runs in practice within limited runtime. This solver provides strategies that reduce quantum space when classical space is added. Moreover, with this solver we can find Pareto fronts for trade-offs between quantum space, classical space and time used to execute a circuit.

\paragraph{Further research} An idea for further research is to build more optimization heuristics for spooky pebble strategies. Currently used optimization strategies emphasize the reduction of ghosts, such that many optimized solutions use (almost) no ghosts. For example optimizers reducing the number of pebbles but increasing the number of ghosts can be developed to better optimize quantum space.

It would also be interesting to study other solving techniques for the spooky pebble game. The recently developed Incremental Property Directed Reachability~\cite{blankestijn2023incremental} and divide-and-conquer methods~\cite{zhang2025divide} would be worthwhile to apply to spooky pebble game solving, as they have promising results for solving the reversible pebble game. 

Another idea for further research is to investigate in which contexts the spooky pebble game and measurement-based uncomputation can be used. Our main motivation in this paper to study these concepts was simulation of classical algorithms on a superposition of inputs on a quantum computer. It would be interesting to investigate the extent to which these techniques can also be used for broader classes of quantum algorithms. 

Finally, it would be worthwhile to investigate other quantum space reduction techniques and combine them with spooky pebble games. A promising candidate is Quantum Catalytic Space~\cite{buhrman2025quantum}. Catalytic space has already had great impact on the classical setting~\cite{cook2024tree}, so developing a pebble game formulation for the quantum setting and integrating it into the spooky pebbling framework would be a promising next step.

\paragraph{Acknowledgements.} This publication is part of the project Divide \& Quantum  (with project number 1389.20.241) of the research programme NWA-ORC which is partly financed by the Dutch Research Council (NWO).

\bibliographystyle{alphaurl}
\bibliography{Bibliography.bib}

\end{document}

%% file: Pebble_moves.tex
    \newcommand{\hdist}{1.4}    %
    \newcommand{\vdist}{0.8}    %

    \newcommand{\1}{0}  %
    \newcommand{\2}{1}
    \newcommand{\3}{1}
    \newcommand{\4}{1}
    \newcommand{\5}{1}
    \newcommand{\6}{0}
    \newcommand{\hshift}{0}
    \newcommand{\vshift}{0}
    \input{DAG_example.tex}

    \renewcommand{\1}{0}
    \renewcommand{\2}{1}
    \renewcommand{\3}{1}
    \renewcommand{\4}{0}
    \renewcommand{\5}{1}
    \renewcommand{\6}{0}
    \renewcommand{\hshift}{-5.5}
    \renewcommand{\vshift}{-5.5}
    \input{DAG_example.tex}

    \renewcommand{\1}{0}
    \renewcommand{\2}{1}
    \renewcommand{\3}{1}
    \renewcommand{\4}{1}
    \renewcommand{\5}{0}
    \renewcommand{\6}{0}
    \renewcommand{\hshift}{0}
    \renewcommand{\vshift}{-5.5}
    \input{DAG_example.tex}

    \renewcommand{\1}{0}
    \renewcommand{\2}{1}
    \renewcommand{\3}{1}
    \renewcommand{\4}{2}
    \renewcommand{\5}{1}
    \renewcommand{\6}{0}
    \renewcommand{\hshift}{5.5}
    \renewcommand{\vshift}{-5.5}
    \input{DAG_example.tex}

    \draw [->,very thick] (.6*\hdist,-.8*\vdist) -- node[text width = 2 cm, align = center, fill=white] {\emph{unpebble(D)}\\  irreversible} (-2.9*\hdist,-4.2*\vdist);
    \draw [->,very thick] (1*\hdist,-.8*\vdist) -- node[text width = 2 cm, align = center, fill=white] {\emph{unpebble(E)}\\ irreversible,\\ reversible,\\ spooky} (1*\hdist,-4.2*\vdist);
    \draw [->,very thick] (1.4*\hdist,-.8*\vdist) -- node[text width = 2 cm, align = center, fill=white] {\emph{ghost(D)}\\ spooky} (4.9*\hdist,-4.2*\vdist);
    

%% file: DAG_example.tex
\Colors{\1}
\node[shape=circle,draw=black,text=\TextColor,fill=\FillColor] (1) at (0+\hshift,0+\vshift) {A};

\Colors{\2}
\node[shape=circle,draw=black,text=\TextColor,fill=\FillColor] (2) at (\hdist+\hshift,0+\vshift) {B};

\Colors{\3}
\node[shape=circle,draw=black,text=\TextColor,fill=\FillColor] (3) at (2*\hdist+\hshift,0+\vshift) {C};

\Colors{\4}
\node[shape=circle,draw=black,text=\TextColor,fill=\FillColor] (4) at (0.5*\hdist+\hshift,\vdist+\vshift) {D};

\Colors{\5}
\node[shape=circle,draw=black,text=\TextColor,fill=\FillColor] (5) at (1.5*\hdist+\hshift,\vdist+\vshift) {E};

\Colors{\6}
\node[shape=circle,draw=black,text=\TextColor,fill=\FillColor] (6) at (\hdist+\hshift,2*\vdist+\vshift) {F} ;

\path [->](1) edge (4);
\path [->](2) edge (4);
\path [->](2) edge (5);
\path [->](3) edge (5);
\path [->](4) edge (6);
\path [->](5) edge (6); 

%% file: Mmbased-uncomp-moves.tex
        \newcommand{\hdist}{1.4}    %
        \newcommand{\vdist}{0.8}    %
        
        \newcommand{\1}{0}  %
        \newcommand{\2}{1}
        \newcommand{\3}{1}
        \newcommand{\lbl}{A}
        \newcommand{\hshift}{0}
        \newcommand{\vshift}{0}
        
        \newcommand{\hpicShift}{2*\hdist}
        \newcommand{\vpicShift}{-2.5*\vdist}

        \renewcommand{\1}{1}
        \renewcommand{\2}{1}
        \renewcommand{\3}{0}
        \renewcommand{\lbl}{1}
        \renewcommand{\hshift}{0}
        \renewcommand{\vshift}{0}
        \input{DAG_example_uncomp.tex}

        \renewcommand{\1}{1}
        \renewcommand{\2}{1}
        \renewcommand{\3}{1}
        \renewcommand{\lbl}{2}
        \renewcommand{\hshift}{\hpicShift}
        \renewcommand{\vshift}{0}
        \input{DAG_example_uncomp.tex}

        \renewcommand{\1}{0}
        \renewcommand{\2}{1}
        \renewcommand{\3}{1}
        \renewcommand{\lbl}{3}
        \renewcommand{\hshift}{2*\hpicShift}
        \renewcommand{\vshift}{0}
        \input{DAG_example_uncomp.tex}

        \renewcommand{\1}{0}
        \renewcommand{\2}{1}
        \renewcommand{\3}{2}
        \renewcommand{\lbl}{4}
        \renewcommand{\hshift}{3*\hpicShift}
        \renewcommand{\vshift}{0}
        \input{DAG_example_uncomp.tex}

        \renewcommand{\1}{1}
        \renewcommand{\2}{1}
        \renewcommand{\3}{2}
        \renewcommand{\lbl}{5}
        \renewcommand{\hshift}{0*\hpicShift}
        \renewcommand{\vshift}{\vpicShift}
        \input{DAG_example_uncomp.tex}

        \renewcommand{\1}{1}
        \renewcommand{\2}{1}
        \renewcommand{\3}{1}
        \renewcommand{\lbl}{6}
        \renewcommand{\hshift}{\hpicShift}
        \renewcommand{\vshift}{\vpicShift}
        \input{DAG_example_uncomp.tex}

        \renewcommand{\1}{1}
        \renewcommand{\2}{1}
        \renewcommand{\3}{0}
        \renewcommand{\lbl}{7}
        \renewcommand{\hshift}{2*\hpicShift}
        \renewcommand{\vshift}{\vpicShift}
        \input{DAG_example_uncomp.tex}

%% file: compiled-circuit-rev.tex
\begin{tikzpicture}[scale=.6,plusss/.style={circle,path picture={ 
                        \draw[black]
                        (path picture bounding box.south) -- (path picture bounding box.north) (path picture bounding box.west) -- (path picture bounding box.east);
    }}]

		\node [style=none] (0) at (0, 0) {};
		\node [style=none] (1) at (0.75, 0) {};
		\node [style=none] (2) at (0, -1) {};
		\node [style=none] (3) at (0.75, -1) {};
		\node [style=none] (4) at (0.75, 0.25) {};
		\node [style=none] (6) at (1.25, -1.25) {};
		\node [style=none] (7) at (0.75, -1.25) {};
		\node [style=none] (8) at (1, -1.25) {};
		\node [style=none] (9) at (0, -2) {};
		\node [style=none] (10) at (0, -3) {};
		\node [style=none] (11) at (1, -3) {};
		\node [style=none] (12) at (1.75, -2) {};
		\node [style=none] (13) at (2.25, -2) {};
		\node [style=none] (14) at (1.75, -1) {};
		\node [style=none] (15) at (2.25, -1) {};
		\node [style=none] (16) at (1.25, 0) {};
		\node [style=none] (17) at (1.25, -1) {};
		\node [style=none] (18) at (1.75, -0.75) {};
		\node [style=none] (19) at (2.25, -0.75) {};
		\node [style=none] (20) at (1.75, -2.25) {};
		\node [style=none] (21) at (2.25, -2.25) {};
		\node [style=none] (22) at (2, -2.25) {};
		\node [style=none] (23) at (0, -4) {};
		\node [style=none] (24) at (2, -4) {};
		\node [style=none] (25) at (2.75, -1) {};
		\node [style=none] (26) at (2.75, -0.75) {};
		\node [style=none] (27) at (3.25, -0.75) {};
		\node [style=none] (28) at (3.25, -1) {};
		\node [style=none] (29) at (2.75, -3) {};
		\node [style=none] (30) at (3.25, -3) {};
		\node [style=none] (31) at (2.75, -3.25) {};
		\node [style=none] (32) at (3, -3.25) {};
		\node [style=none] (33) at (3.25, -3.25) {};
		\node [style=none] (34) at (3, -5) {};
		\node [style=none] (36) at (0, -5) {};
		\node [style=none] (37) at (4, -6) {};
		\node [style=none] (38) at (3.75, -0.75) {};
		\node [style=none] (39) at (4.25, -0.75) {};
		\node [style=none] (40) at (3.75, -1) {};
		\node [style=none] (41) at (4.25, -1) {};
		\node [style=none] (42) at (3.75, -4) {};
		\node [style=none] (43) at (4.25, -4) {};
		\node [style=none] (44) at (4.25, -4.25) {};
		\node [style=none] (45) at (3.75, -4.25) {};
		\node [style=none] (46) at (4, -4.25) {};
		\node [style=none] (47) at (0, -6) {};
		\node [style=none] (48) at (4, -6) {};
		\node [style=none] (49) at (4, -6) {};
		\node [style=none] (50) at (3.75, -3) {};
		\node [style=none] (51) at (4.25, -3) {};
		\node [style=none] (52) at (2.75, -2) {};
		\node [style=none] (53) at (3.25, -2) {};
		\node [style=none] (54) at (3.75, -2) {};
		\node [style=none] (55) at (4.25, -2) {};
		\node [style=none] (56) at (4.75, 0) {};
		\node [style=none] (57) at (5.25, 0) {};
		\node [style=none] (58) at (5.25, 0.25) {};
		\node [style=none] (59) at (4.75, 0.25) {};
		\node [style=none] (60) at (4.75, -1) {};
		\node [style=none] (61) at (4.75, -1.25) {};
		\node [style=none] (62) at (5.25, -1.25) {};
		\node [style=none] (63) at (5.25, -1) {};
		\node [style=none] (64) at (5, -1.25) {};
		\node [style=none] (65) at (5, -3) {};
		\node [style=none] (66) at (6, -3) {};
		\node [style=none] (67) at (5.75, -6.25) {};
		\node [style=none] (68) at (6.25, -6.25) {};
		\node [style=none] (69) at (5.75, -5) {};
		\node [style=none] (70) at (6.25, -5) {};
		\node [style=none] (71) at (5.75, -6) {};
		\node [style=none] (72) at (6.25, -6) {};
		\node [style=none] (73) at (5.75, -4.75) {};
		\node [style=none] (74) at (6.25, -4.75) {};
		\node [style=none] (75) at (6, -4.75) {};
		\node [style=none] (76) at (7, -6) {};
		\node [style=none] (77) at (6.75, -0.75) {};
		\node [style=none] (78) at (7.25, -0.75) {};
		\node [style=none] (79) at (6.75, -1) {};
		\node [style=none] (80) at (7.25, -1) {};
		\node [style=none] (81) at (6.75, -4) {};
		\node [style=none] (82) at (7.25, -4) {};
		\node [style=none] (83) at (7.25, -4.25) {};
		\node [style=none] (84) at (6.75, -4.25) {};
		\node [style=none] (85) at (7, -4.25) {};
		\node [style=none] (86) at (7, -6) {};
		\node [style=none] (87) at (7, -6) {};
		\node [style=none] (88) at (6.75, -3) {};
		\node [style=none] (89) at (7.25, -3) {};
		\node [style=none] (90) at (6.75, -2) {};
		\node [style=none] (91) at (7.25, -2) {};
		\node [style=none] (92) at (7.75, -2) {};
		\node [style=none] (93) at (8.25, -2) {};
		\node [style=none] (94) at (7.75, -1) {};
		\node [style=none] (95) at (8.25, -1) {};
		\node [style=none] (96) at (7.75, -0.75) {};
		\node [style=none] (97) at (8.25, -0.75) {};
		\node [style=none] (98) at (7.75, -2.25) {};
		\node [style=none] (99) at (8.25, -2.25) {};
		\node [style=none] (100) at (8, -2.25) {};
		\node [style=none] (101) at (8, -4) {};
		\node [style=none] (102) at (8.75, -1) {};
		\node [style=none] (103) at (9.25, -1) {};
		\node [style=none] (104) at (8.75, 0) {};
		\node [style=none] (105) at (9.25, 0) {};
		\node [style=none] (106) at (8.75, 0.25) {};
		\node [style=none] (107) at (9.25, 0.25) {};
		\node [style=none] (108) at (8.75, -1.25) {};
		\node [style=none] (109) at (9.25, -1.25) {};
		\node [style=none] (110) at (9, -1.25) {};
		\node [style=none] (111) at (9, -4) {};
		\node [style=none] (112) at (9.75, -1) {};
		\node [style=none] (113) at (9.75, -0.75) {};
		\node [style=none] (114) at (10.25, -0.75) {};
		\node [style=none] (115) at (10.25, -1) {};
		\node [style=none] (116) at (9.75, -3) {};
		\node [style=none] (117) at (10.25, -3) {};
		\node [style=none] (118) at (9.75, -3.25) {};
		\node [style=none] (119) at (10, -3.25) {};
		\node [style=none] (120) at (10.25, -3.25) {};
		\node [style=none] (121) at (10, -5) {};
		\node [style=none] (122) at (9.75, -2) {};
		\node [style=none] (123) at (10.25, -2) {};
		\node [style=none] (124) at (10.75, -1) {};
		\node [style=none] (125) at (11.25, -1) {};
		\node [style=none] (126) at (10.75, 0) {};
		\node [style=none] (127) at (11.25, 0) {};
		\node [style=none] (128) at (10.75, 0.25) {};
		\node [style=none] (129) at (11.25, 0.25) {};
		\node [style=none] (130) at (10.75, -1.25) {};
		\node [style=none] (131) at (11.25, -1.25) {};
		\node [style=none] (132) at (11, -1.25) {};
		\node [style=none] (133) at (11, -4) {};
		\node [style=none] (134) at (12, 0) {};
		\node [style=none] (135) at (12, -1) {};
		\node [style=none] (136) at (12, -2) {};
		\node [style=none] (137) at (12, -3) {};
		\node [style=none] (138) at (12, -4) {};
		\node [style=none] (139) at (12, -5) {};
		\node [style=none] (140) at (12, -6) {};
		\node [style=none] (5) at (1.25, 0.25) {};
		\node [style=plus] (141) at (1, -3) {};
		\node [style=plus] (142) at (2, -4) {};
		\node [style=plus] (143) at (3, -5) {};
		\node [style=plus] (144) at (4, -6) {};
		\node [style=plus] (145) at (6, -3) {};
		\node [style=plus] (146) at (5, -3) {};
		\node [style=plus] (147) at (7, -6) {};
		\node [style=plus] (148) at (8, -4) {};
		\node [style=plus] (149) at (9, -4) {};
		\node [style=plus] (150) at (10, -5) {};
		\node [style=plus] (151) at (11, -4) {};
		\node [style=none] (152) at (0, 0) {};
		\node [style=new style 0] (153) at (-0.5, 0) {$x_1$};
		\node [style=new style 0] (154) at (-0.5, -1) {$x_2$};
		\node [style=new style 0] (155) at (-0.5, -2) {$x_3$};
		\node [style=new style 0] (156) at (-0.5, -3) {$|0\rangle$};
		\node [style=new style 0] (157) at (-0.5, -4) {$|0\rangle$};
		\node [style=new style 0] (158) at (-0.5, -5) {$|0\rangle$};
		\node [style=new style 0] (159) at (-0.5, -6) {$|0\rangle$};
		\node [style=new style 0] (160) at (12.5, -6) {$|0\rangle$};
		\node [style=new style 0] (161) at (12.5, -5) {$|0\rangle$};
		\node [style=new style 0] (162) at (12.5, -4) {$|0\rangle$};
		\node [style=new style 0] (163) at (12.5, -3) {$E$};
		\node [style=new style 0] (164) at (12.5, -2) {$x_3$};
		\node [style=new style 0] (165) at (12.5, -1) {$x_2$};
		\node [style=new style 0] (166) at (12.5, 0) {$x_1$};
		\node [style=new style 0] (167) at (-0.5, 0) {};
		\node [style=new style 0] (168) at (-0.5, 0) {};
		\node [style=new style 0] (169) at (-0.5, 0) {};
		\node [style=new style 0] (170) at (-0.5, 0) {};
		\draw (1.center) to (0.center);
		\draw (2.center) to (3.center);
		\draw (4.center) to (7.center);
		\draw (7.center) to (6.center);
		\draw (5.center) to (6.center);
		\draw (5.center) to (4.center);
		\draw (8.center) to (11.center);
		\draw (17.center) to (14.center);
		\draw (9.center) to (12.center);
		\draw (18.center) to (20.center);
		\draw (20.center) to (21.center);
		\draw (21.center) to (19.center);
		\draw (19.center) to (18.center);
		\draw (22.center) to (24.center);
		\draw (25.center) to (15.center);
		\draw (26.center) to (31.center);
		\draw (26.center) to (27.center);
		\draw (27.center) to (33.center);
		\draw (31.center) to (33.center);
		\draw (32.center) to (34.center);
		\draw (38.center) to (45.center);
		\draw (45.center) to (44.center);
		\draw (44.center) to (39.center);
		\draw (38.center) to (39.center);
		\draw (28.center) to (40.center);
		\draw (46.center) to (48.center);
		\draw (23.center) to (24.center);
		\draw (36.center) to (34.center);
		\draw (47.center) to (48.center);
		\draw (10.center) to (11.center);
		\draw (59.center) to (61.center);
		\draw (59.center) to (58.center);
		\draw (58.center) to (62.center);
		\draw (61.center) to (62.center);
		\draw (16.center) to (56.center);
		\draw (64.center) to (65.center);
		\draw (41.center) to (60.center);
		\draw (13.center) to (52.center);
		\draw (53.center) to (54.center);
		\draw (65.center) to (66.center);
		\draw (66.center) to (75.center);
		\draw (73.center) to (67.center);
		\draw (67.center) to (68.center);
		\draw (68.center) to (74.center);
		\draw (73.center) to (74.center);
		\draw (77.center) to (84.center);
		\draw (84.center) to (83.center);
		\draw (83.center) to (78.center);
		\draw (77.center) to (78.center);
		\draw (85.center) to (86.center);
		\draw (63.center) to (79.center);
		\draw (55.center) to (90.center);
		\draw (96.center) to (98.center);
		\draw (98.center) to (99.center);
		\draw (99.center) to (97.center);
		\draw (97.center) to (96.center);
		\draw (100.center) to (101.center);
		\draw (80.center) to (94.center);
		\draw (91.center) to (92.center);
		\draw (106.center) to (108.center);
		\draw (108.center) to (109.center);
		\draw (109.center) to (107.center);
		\draw (107.center) to (106.center);
		\draw (110.center) to (111.center);
		\draw (101.center) to (111.center);
		\draw (113.center) to (118.center);
		\draw (113.center) to (114.center);
		\draw (114.center) to (120.center);
		\draw (118.center) to (120.center);
		\draw (119.center) to (121.center);
		\draw (103.center) to (112.center);
		\draw (95.center) to (102.center);
		\draw (93.center) to (122.center);
		\draw (128.center) to (130.center);
		\draw (130.center) to (131.center);
		\draw (131.center) to (129.center);
		\draw (129.center) to (128.center);
		\draw (132.center) to (133.center);
		\draw (115.center) to (124.center);
		\draw (57.center) to (104.center);
		\draw (105.center) to (126.center);
		\draw (125.center) to (135.center);
		\draw (127.center) to (134.center);
		\draw (123.center) to (136.center);
		\draw (133.center) to (138.center);
		\draw (121.center) to (139.center);
		\draw (140.center) to (87.center);
		\draw [style=red] (49.center) to (71.center);
		\draw [style=red] (72.center) to (87.center);
		\draw [style=red] (70.center) to (121.center);
		\draw [style=red] (34.center) to (69.center);
		\draw [style=red] (24.center) to (42.center);
		\draw [style=red] (43.center) to (81.center);
		\draw [style=red] (82.center) to (101.center);
		\draw [style=red] (111.center) to (133.center);
		\draw [style=red] (116.center) to (89.center);
		\draw [style=red] (117.center) to (137.center);
		\draw [style=red] (88.center) to (66.center);
		\draw [style=red] (65.center) to (51.center);
		\draw [style=red] (50.center) to (30.center);
		\draw [style=red] (29.center) to (11.center);
		\draw [style=dtt] (52.center) to (53.center);
		\draw [style=dtt] (54.center) to (55.center);
		\draw [style=dtt] (90.center) to (91.center);
		\draw [style=dtt] (122.center) to (123.center);
		\draw [style=red dtt] (50.center) to (51.center);
		\draw [style=red dtt] (88.center) to (89.center);

        \node [style=none] (0) at (0, 0) {};
		\node [style=none] (1) at (0.75, 0) {};
		\node [style=none] (2) at (0, -1) {};
		\node [style=none] (3) at (0.75, -1) {};
		\node [style=none] (4) at (0.75, 0.25) {};
		\node [style=none] (6) at (1.25, -1.25) {};
		\node [style=none] (7) at (0.75, -1.25) {};
		\node [style=none] (8) at (1, -1.25) {};
		\node [style=none] (9) at (0, -2) {};
		\node [style=none] (10) at (0, -3) {};
		\node [style=none] (11) at (1, -3) {};
		\node [style=none] (12) at (1.75, -2) {};
		\node [style=none] (13) at (2.25, -2) {};
		\node [style=none] (14) at (1.75, -1) {};
		\node [style=none] (15) at (2.25, -1) {};
		\node [style=none] (16) at (1.25, 0) {};
		\node [style=none] (17) at (1.25, -1) {};
		\node [style=none] (18) at (1.75, -0.75) {};
		\node [style=none] (19) at (2.25, -0.75) {};
		\node [style=none] (20) at (1.75, -2.25) {};
		\node [style=none] (21) at (2.25, -2.25) {};
		\node [style=none] (22) at (2, -2.25) {};
		\node [style=none] (23) at (0, -4) {};
		\node [style=none] (24) at (2, -4) {};
		\node [style=none] (25) at (2.75, -1) {};
		\node [style=none] (26) at (2.75, -0.75) {};
		\node [style=none] (27) at (3.25, -0.75) {};
		\node [style=none] (28) at (3.25, -1) {};
		\node [style=none] (29) at (2.75, -3) {};
		\node [style=none] (30) at (3.25, -3) {};
		\node [style=none] (31) at (2.75, -3.25) {};
		\node [style=none] (32) at (3, -3.25) {};
		\node [style=none] (33) at (3.25, -3.25) {};
		\node [style=none] (34) at (3, -5) {};
		\node [style=none] (36) at (0, -5) {};
		\node [style=none] (37) at (4, -6) {};
		\node [style=none] (38) at (3.75, -0.75) {};
		\node [style=none] (39) at (4.25, -0.75) {};
		\node [style=none] (40) at (3.75, -1) {};
		\node [style=none] (41) at (4.25, -1) {};
		\node [style=none] (42) at (3.75, -4) {};
		\node [style=none] (43) at (4.25, -4) {};
		\node [style=none] (44) at (4.25, -4.25) {};
		\node [style=none] (45) at (3.75, -4.25) {};
		\node [style=none] (46) at (4, -4.25) {};
		\node [style=none] (47) at (0, -6) {};
		\node [style=none] (48) at (4, -6) {};
		\node [style=none] (49) at (4, -6) {};
		\node [style=none] (50) at (3.75, -3) {};
		\node [style=none] (51) at (4.25, -3) {};
		\node [style=none] (52) at (2.75, -2) {};
		\node [style=none] (53) at (3.25, -2) {};
		\node [style=none] (54) at (3.75, -2) {};
		\node [style=none] (55) at (4.25, -2) {};
		\node [style=none] (56) at (4.75, 0) {};
		\node [style=none] (57) at (5.25, 0) {};
		\node [style=none] (58) at (5.25, 0.25) {};
		\node [style=none] (59) at (4.75, 0.25) {};
		\node [style=none] (60) at (4.75, -1) {};
		\node [style=none] (61) at (4.75, -1.25) {};
		\node [style=none] (62) at (5.25, -1.25) {};
		\node [style=none] (63) at (5.25, -1) {};
		\node [style=none] (64) at (5, -1.25) {};
		\node [style=none] (65) at (5, -3) {};
		\node [style=none] (66) at (6, -3) {};
		\node [style=none] (67) at (5.75, -6.25) {};
		\node [style=none] (68) at (6.25, -6.25) {};
		\node [style=none] (69) at (5.75, -5) {};
		\node [style=none] (70) at (6.25, -5) {};
		\node [style=none] (71) at (5.75, -6) {};
		\node [style=none] (72) at (6.25, -6) {};
		\node [style=none] (73) at (5.75, -4.75) {};
		\node [style=none] (74) at (6.25, -4.75) {};
		\node [style=none] (75) at (6, -4.75) {};
		\node [style=none] (76) at (7, -6) {};
		\node [style=none] (77) at (6.75, -0.75) {};
		\node [style=none] (78) at (7.25, -0.75) {};
		\node [style=none] (79) at (6.75, -1) {};
		\node [style=none] (80) at (7.25, -1) {};
		\node [style=none] (81) at (6.75, -4) {};
		\node [style=none] (82) at (7.25, -4) {};
		\node [style=none] (83) at (7.25, -4.25) {};
		\node [style=none] (84) at (6.75, -4.25) {};
		\node [style=none] (85) at (7, -4.25) {};
		\node [style=none] (86) at (7, -6) {};
		\node [style=none] (87) at (7, -6) {};
		\node [style=none] (88) at (6.75, -3) {};
		\node [style=none] (89) at (7.25, -3) {};
		\node [style=none] (90) at (6.75, -2) {};
		\node [style=none] (91) at (7.25, -2) {};
		\node [style=none] (92) at (7.75, -2) {};
		\node [style=none] (93) at (8.25, -2) {};
		\node [style=none] (94) at (7.75, -1) {};
		\node [style=none] (95) at (8.25, -1) {};
		\node [style=none] (96) at (7.75, -0.75) {};
		\node [style=none] (97) at (8.25, -0.75) {};
		\node [style=none] (98) at (7.75, -2.25) {};
		\node [style=none] (99) at (8.25, -2.25) {};
		\node [style=none] (100) at (8, -2.25) {};
		\node [style=none] (101) at (8, -4) {};
		\node [style=none] (102) at (8.75, -1) {};
		\node [style=none] (103) at (9.25, -1) {};
		\node [style=none] (104) at (8.75, 0) {};
		\node [style=none] (105) at (9.25, 0) {};
		\node [style=none] (106) at (8.75, 0.25) {};
		\node [style=none] (107) at (9.25, 0.25) {};
		\node [style=none] (108) at (8.75, -1.25) {};
		\node [style=none] (109) at (9.25, -1.25) {};
		\node [style=none] (110) at (9, -1.25) {};
		\node [style=none] (111) at (9, -4) {};
		\node [style=none] (112) at (9.75, -1) {};
		\node [style=none] (113) at (9.75, -0.75) {};
		\node [style=none] (114) at (10.25, -0.75) {};
		\node [style=none] (115) at (10.25, -1) {};
		\node [style=none] (116) at (9.75, -3) {};
		\node [style=none] (117) at (10.25, -3) {};
		\node [style=none] (118) at (9.75, -3.25) {};
		\node [style=none] (119) at (10, -3.25) {};
		\node [style=none] (120) at (10.25, -3.25) {};
		\node [style=none] (121) at (10, -5) {};
		\node [style=none] (122) at (9.75, -2) {};
		\node [style=none] (123) at (10.25, -2) {};
		\node [style=none] (124) at (10.75, -1) {};
		\node [style=none] (125) at (11.25, -1) {};
		\node [style=none] (126) at (10.75, 0) {};
		\node [style=none] (127) at (11.25, 0) {};
		\node [style=none] (128) at (10.75, 0.25) {};
		\node [style=none] (129) at (11.25, 0.25) {};
		\node [style=none] (130) at (10.75, -1.25) {};
		\node [style=none] (131) at (11.25, -1.25) {};
		\node [style=none] (132) at (11, -1.25) {};
		\node [style=none] (133) at (11, -4) {};
		\node [style=none] (134) at (12, 0) {};
		\node [style=none] (135) at (12, -1) {};
		\node [style=none] (136) at (12, -2) {};
		\node [style=none] (137) at (12, -3) {};
		\node [style=none] (138) at (12, -4) {};
		\node [style=none] (139) at (12, -5) {};
		\node [style=none] (140) at (12, -6) {};
		\node [style=none] (5) at (1.25, 0.25) {};
		\node [style=plus] (141) at (1, -3) {};
		\node [style=plus] (142) at (2, -4) {};
		\node [style=plus] (143) at (3, -5) {};
		\node [style=plus] (144) at (4, -6) {};
		\node [style=plus] (145) at (6, -3) {};
		\node [style=plus] (146) at (5, -3) {};
		\node [style=plus] (147) at (7, -6) {};
		\node [style=plus] (148) at (8, -4) {};
		\node [style=plus] (149) at (9, -4) {};
		\node [style=plus] (150) at (10, -5) {};
		\node [style=plus] (151) at (11, -4) {};
		\node [style=none] (152) at (0, 0) {};
		\node [style=new style 0] (153) at (-0.5, 0) {$x_1$};
		\node [style=new style 0] (154) at (-0.5, -1) {$x_2$};
		\node [style=new style 0] (155) at (-0.5, -2) {$x_3$};
		\node [style=new style 0] (156) at (-0.5, -3) {$|0\rangle$};
		\node [style=new style 0] (157) at (-0.5, -4) {$|0\rangle$};
		\node [style=new style 0] (158) at (-0.5, -5) {$|0\rangle$};
		\node [style=new style 0] (159) at (-0.5, -6) {$|0\rangle$};
		\node [style=new style 0] (160) at (12.5, -6) {$|0\rangle$};
		\node [style=new style 0] (161) at (12.5, -5) {$|0\rangle$};
		\node [style=new style 0] (162) at (12.5, -4) {$|0\rangle$};
		\node [style=new style 0] (163) at (12.5, -3) {$E$};
		\node [style=new style 0] (164) at (12.5, -2) {$x_3$};
		\node [style=new style 0] (165) at (12.5, -1) {$x_2$};
		\node [style=new style 0] (166) at (12.5, 0) {$x_1$};
		\node [style=new style 0] (167) at (-0.5, 0) {};
		\node [style=new style 0] (168) at (-0.5, 0) {};
		\node [style=new style 0] (169) at (-0.5, 0) {};
		\node [style=new style 0] (170) at (-0.5, 0) {};
  
        \node [style=none] (171) at (1, -0.5) {A};
		\node [style=none] (172) at (2, -1.5) {B};
		\node [style=none] (173) at (3, -2) {C};
		\node [style=none] (174) at (4, -2.5) {D};
		\node [style=none] (175) at (5, -0.5) {A};
		\node [style=none] (176) at (6, -5.5) {E};
		\node [style=none] (177) at (7, -2.5) {D};
		\node [style=none] (178) at (8, -1.5) {B};
		\node [style=none] (179) at (9, -0.5) {A};
		\node [style=none] (180) at (10, -2) {C};
		\node [style=none] (181) at (11, -0.5) {A};
    \end{tikzpicture}

%% file: compiled-circuit-spooky.tex
\begin{tikzpicture}[scale=.6,plusss/.style={circle,path picture={ 
                        \draw[black]
                        (path picture bounding box.south) -- (path picture bounding box.north) (path picture bounding box.west) -- (path picture bounding box.east);
    }}]

		\node [style=none] (0) at (0, 0) {};
		\node [style=none] (1) at (0, -1) {};
		\node [style=none] (2) at (0, -2) {};
		\node [style=none] (3) at (0, -3) {};
		\node [style=none] (4) at (0, -4) {};
		\node [style=none] (5) at (0, -5) {};
		\node [style=none] (6) at (0.75, 0) {};
		\node [style=none] (7) at (0.75, -1) {};
		\node [style=none] (8) at (0.75, 0.25) {};
		\node [style=none] (9) at (1.25, -1.25) {};
		\node [style=none] (10) at (0.75, -1.25) {};
		\node [style=none] (11) at (1, -1.25) {};
		\node [style=none] (12) at (1, -3) {};
		\node [style=none] (13) at (1.25, 0) {};
		\node [style=none] (14) at (1.25, -1) {};
		\node [style=none] (15) at (1.25, 0.25) {};
		\node [style=none] (17) at (1.75, -1) {};
		\node [style=none] (18) at (1.75, -0.75) {};
		\node [style=none] (19) at (2.25, -0.75) {};
		\node [style=none] (20) at (2.25, -1) {};
		\node [style=none] (21) at (1.75, -3) {};
		\node [style=none] (22) at (2.25, -3) {};
		\node [style=none] (23) at (1.75, -3.25) {};
		\node [style=none] (24) at (2, -3.25) {};
		\node [style=none] (25) at (2.25, -3.25) {};
		\node [style=none] (26) at (2, -4) {};
		\node [style=none] (27) at (1.75, -2) {};
		\node [style=none] (28) at (2.25, -2) {};
		\node [style=none] (29) at (2.75, 0) {};
		\node [style=none] (30) at (2.75, -1) {};
		\node [style=none] (31) at (2.75, 0.25) {};
		\node [style=none] (32) at (3.25, -1.25) {};
		\node [style=none] (33) at (2.75, -1.25) {};
		\node [style=none] (34) at (3, -1.25) {};
		\node [style=none] (35) at (3, -3) {};
		\node [style=none] (36) at (3.25, 0) {};
		\node [style=none] (37) at (3.25, -1) {};
		\node [style=none] (38) at (3.25, 0.25) {};
		\node [style=none] (39) at (3.75, -1) {};
		\node [style=none] (40) at (3.75, -2) {};
		\node [style=none] (41) at (3.75, -0.75) {};
		\node [style=none] (42) at (4.25, -2.25) {};
		\node [style=none] (43) at (3.75, -2.25) {};
		\node [style=none] (44) at (4, -2.25) {};
		\node [style=none] (45) at (4, -3) {};
		\node [style=none] (46) at (4.25, -1) {};
		\node [style=none] (47) at (4.25, -2) {};
		\node [style=none] (48) at (4.25, -0.75) {};
		\node [style=none] (49) at (4.75, -1) {};
		\node [style=none] (50) at (4.75, -0.75) {};
		\node [style=none] (51) at (5.25, -0.75) {};
		\node [style=none] (52) at (5.25, -1) {};
		\node [style=none] (53) at (4.75, -3) {};
		\node [style=none] (54) at (5.25, -3) {};
		\node [style=none] (55) at (4.75, -3.25) {};
		\node [style=none] (56) at (5, -3.25) {};
		\node [style=none] (57) at (5.25, -3.25) {};
		\node [style=none] (58) at (5, -5) {};
		\node [style=none] (59) at (4.75, -2) {};
		\node [style=none] (60) at (5.25, -2) {};
		\node [style=none] (61) at (5.75, -1) {};
		\node [style=none] (62) at (5.75, -2) {};
		\node [style=none] (63) at (5.75, -0.75) {};
		\node [style=none] (64) at (6.25, -2.25) {};
		\node [style=none] (65) at (5.75, -2.25) {};
		\node [style=none] (66) at (6, -2.25) {};
		\node [style=none] (67) at (6, -3) {};
		\node [style=none] (68) at (6.25, -1) {};
		\node [style=none] (69) at (6.25, -2) {};
		\node [style=none] (70) at (6.25, -0.75) {};
		\node [style=none] (71) at (7, -3) {};
		\node [style=none] (72) at (6.75, -5.25) {};
		\node [style=none] (73) at (7.25, -5.25) {};
		\node [style=none] (74) at (6.75, -4) {};
		\node [style=none] (75) at (7.25, -4) {};
		\node [style=none] (76) at (6.75, -5) {};
		\node [style=none] (77) at (7.25, -5) {};
		\node [style=none] (78) at (6.75, -3.75) {};
		\node [style=none] (79) at (7.25, -3.75) {};
		\node [style=none] (80) at (7, -3.75) {};
		\node [style=none] (81) at (7.75, -3.75) {};
		\node [style=none] (82) at (8.25, -3.75) {};
		\node [style=none] (83) at (8.25, -4.25) {};
		\node [style=none] (84) at (7.75, -4.25) {};
		\node [style=none] (85) at (8.75, -1) {};
		\node [style=none] (86) at (8.75, -2) {};
		\node [style=none] (87) at (8.75, -0.75) {};
		\node [style=none] (88) at (9.25, -2.25) {};
		\node [style=none] (89) at (8.75, -2.25) {};
		\node [style=none] (90) at (9, -2.25) {};
		\node [style=none] (91) at (9, -4) {};
		\node [style=none] (92) at (9.25, -1) {};
		\node [style=none] (93) at (9.25, -2) {};
		\node [style=none] (94) at (9.25, -0.75) {};
		\node [style=none] (95) at (9.75, -1) {};
		\node [style=none] (96) at (9.75, -0.75) {};
		\node [style=none] (97) at (10.25, -0.75) {};
		\node [style=none] (98) at (10.25, -1) {};
		\node [style=none] (99) at (9.75, -4) {};
		\node [style=none] (100) at (10.25, -4) {};
		\node [style=none] (101) at (9.75, -4.25) {};
		\node [style=none] (102) at (10, -4.25) {};
		\node [style=none] (103) at (10.25, -4.25) {};
		\node [style=none] (104) at (10, -5) {};
		\node [style=none] (105) at (9.75, -2) {};
		\node [style=none] (106) at (10.25, -2) {};
		\node [style=none] (107) at (10.75, -1) {};
		\node [style=none] (108) at (10.75, -2) {};
		\node [style=none] (109) at (10.75, -0.75) {};
		\node [style=none] (110) at (11.25, -2.25) {};
		\node [style=none] (111) at (10.75, -2.25) {};
		\node [style=none] (112) at (11, -2.25) {};
		\node [style=none] (113) at (11, -4) {};
		\node [style=none] (114) at (11.25, -1) {};
		\node [style=none] (115) at (11.25, -2) {};
		\node [style=none] (116) at (11.25, -0.75) {};
		\node [style=none] (117) at (11.75, 0) {};
		\node [style=none] (118) at (11.75, -1) {};
		\node [style=none] (119) at (11.75, 0.25) {};
		\node [style=none] (120) at (12.25, -1.25) {};
		\node [style=none] (121) at (11.75, -1.25) {};
		\node [style=none] (122) at (12, -1.25) {};
		\node [style=none] (123) at (12, -4) {};
		\node [style=none] (124) at (12.25, 0) {};
		\node [style=none] (125) at (12.25, -1) {};
		\node [style=none] (126) at (12.25, 0.25) {};
		\node [style=none] (127) at (12.75, -1) {};
		\node [style=none] (128) at (12.75, -0.75) {};
		\node [style=none] (129) at (13.25, -0.75) {};
		\node [style=none] (130) at (13.25, -1) {};
		\node [style=none] (131) at (12.75, -4) {};
		\node [style=none] (132) at (13.25, -4) {};
		\node [style=none] (133) at (12.75, -4.25) {};
		\node [style=none] (134) at (13, -4.25) {};
		\node [style=none] (135) at (13.25, -4.25) {};
		\node [style=none] (136) at (13, -5) {};
		\node [style=none] (137) at (12.75, -2) {};
		\node [style=none] (138) at (13.25, -2) {};
		\node [style=none] (139) at (13.75, -4.75) {};
		\node [style=none] (140) at (14.25, -4.75) {};
		\node [style=none] (141) at (13.75, -5.25) {};
		\node [style=none] (142) at (14.25, -5.25) {};
		\node [style=none] (143) at (14.75, -1) {};
		\node [style=none] (144) at (14.75, -0.75) {};
		\node [style=none] (145) at (15.25, -0.75) {};
		\node [style=none] (146) at (15.25, -1) {};
		\node [style=none] (147) at (14.75, -4) {};
		\node [style=none] (148) at (15.25, -4) {};
		\node [style=none] (149) at (14.75, -4.25) {};
		\node [style=none] (150) at (15, -4.25) {};
		\node [style=none] (151) at (15.25, -4.25) {};
		\node [style=none] (152) at (15, -5) {};
		\node [style=none] (153) at (14.75, -2) {};
		\node [style=none] (154) at (15.25, -2) {};
		\node [style=none] (155) at (15.75, 0) {};
		\node [style=none] (156) at (15.75, -1) {};
		\node [style=none] (157) at (15.75, 0.25) {};
		\node [style=none] (158) at (16.25, -1.25) {};
		\node [style=none] (159) at (15.75, -1.25) {};
		\node [style=none] (160) at (16, -1.25) {};
		\node [style=none] (161) at (16, -4) {};
		\node [style=none] (162) at (16.25, 0) {};
		\node [style=none] (163) at (16.25, -1) {};
		\node [style=none] (164) at (16.25, 0.25) {};
		\node [style=none] (165) at (17, 0) {};
		\node [style=none] (166) at (17, -1) {};
		\node [style=none] (167) at (17, -2) {};
		\node [style=none] (168) at (17, -3) {};
		\node [style=none] (169) at (17, -4) {};
		\node [style=none] (170) at (17, -5) {};
		\node [style=none] (171) at (7.75, -4) {};
		\node [style=none] (172) at (8.25, -4) {};
		\node [style=none] (173) at (13.75, -5) {};
		\node [style=none] (174) at (14.25, -5) {};
		\node [style=none] (175) at (9.75, -3) {};
		\node [style=none] (176) at (10.25, -3) {};
		\node [style=none] (177) at (12.75, -3) {};
		\node [style=none] (178) at (13.25, -3) {};
		\node [style=none] (179) at (14.75, -3) {};
		\node [style=none] (180) at (15.25, -3) {};
		\node [style=none] (181) at (8, -4.25) {};
		\node [style=none] (182) at (8, -6) {};
		\node [style=none] (183) at (14, -6) {};
		\node [style=none] (184) at (14, -5.25) {};
		\node [style=new style 0] (185) at (-0.5, 0) {};
		\node [style=new style 0] (186) at (-0.5, -1) {};
		\node [style=new style 0] (187) at (-0.5, -2) {};
		\node [style=new style 0] (188) at (-0.5, -3) {};
		\node [style=new style 0] (189) at (-0.5, -4) {};
		\node [style=new style 0] (190) at (-0.5, -5) {};
		\node [style=new style 0] (191) at (17.5, 0) {};
		\node [style=new style 0] (192) at (17.5, -1) {};
		\node [style=new style 0] (193) at (17.5, -2) {};
		\node [style=new style 0] (194) at (17.5, -3) {};
		\node [style=new style 0] (195) at (17.5, -4) {};
		\node [style=new style 0] (196) at (17.5, -5) {};
		\node [style=none] (197) at (1, -0.5) {};
		\node [style=none] (198) at (2, -2) {};
		\node [style=none] (199) at (3, -0.5) {};
		\node [style=none] (200) at (4, -1.5) {};
		\node [style=none] (201) at (5, -2) {};
		\node [style=none] (202) at (6, -1.5) {};
		\node [style=none] (203) at (7, -4.5) {};
		\node [style=none] (204) at (8, -4) {};
		\node [style=none] (205) at (9, -1.5) {};
		\node [style=none] (206) at (10, -2.5) {};
		\node [style=none] (207) at (11, -1.5) {};
		\node [style=none] (208) at (12, -0.5) {};
		\node [style=none] (209) at (13, -2.5) {};
		\node [style=none] (210) at (14, -5) {};
		\node [style=none] (211) at (15, -2.5) {};
		\node [style=none] (212) at (16, -0.5) {};

		\draw (8.center) to (10.center);
		\draw (10.center) to (9.center);
		\draw (15.center) to (9.center);
		\draw (15.center) to (8.center);
		\draw (11.center) to (12.center);
		\draw (18.center) to (23.center);
		\draw (18.center) to (19.center);
		\draw (19.center) to (25.center);
		\draw (23.center) to (25.center);
		\draw (24.center) to (26.center);
		\draw (31.center) to (33.center);
		\draw (33.center) to (32.center);
		\draw (38.center) to (32.center);
		\draw (38.center) to (31.center);
		\draw (34.center) to (35.center);
		\draw (41.center) to (43.center);
		\draw (43.center) to (42.center);
		\draw (48.center) to (42.center);
		\draw (48.center) to (41.center);
		\draw (44.center) to (45.center);
		\draw (35.center) to (45.center);
		\draw (3.center) to (12.center);
		\draw (1.center) to (7.center);
		\draw (0.center) to (6.center);
		\draw (2.center) to (27.center);
		\draw (13.center) to (29.center);
		\draw (20.center) to (30.center);
		\draw (14.center) to (17.center);
		\draw (4.center) to (26.center);
		\draw (50.center) to (55.center);
		\draw (50.center) to (51.center);
		\draw (51.center) to (57.center);
		\draw (55.center) to (57.center);
		\draw (56.center) to (58.center);
		\draw (63.center) to (65.center);
		\draw (65.center) to (64.center);
		\draw (70.center) to (64.center);
		\draw (70.center) to (63.center);
		\draw (66.center) to (67.center);
		\draw (71.center) to (80.center);
		\draw (78.center) to (72.center);
		\draw (72.center) to (73.center);
		\draw (73.center) to (79.center);
		\draw (78.center) to (79.center);
		\draw (81.center) to (82.center);
		\draw (82.center) to (83.center);
		\draw (81.center) to (84.center);
		\draw (84.center) to (83.center);
		\draw (87.center) to (89.center);
		\draw (89.center) to (88.center);
		\draw (94.center) to (88.center);
		\draw (94.center) to (87.center);
		\draw (90.center) to (91.center);
		\draw (96.center) to (101.center);
		\draw (96.center) to (97.center);
		\draw (97.center) to (103.center);
		\draw (101.center) to (103.center);
		\draw (102.center) to (104.center);
		\draw (109.center) to (111.center);
		\draw (111.center) to (110.center);
		\draw (116.center) to (110.center);
		\draw (116.center) to (109.center);
		\draw (112.center) to (113.center);
		\draw (119.center) to (121.center);
		\draw (121.center) to (120.center);
		\draw (126.center) to (120.center);
		\draw (126.center) to (119.center);
		\draw (122.center) to (123.center);
		\draw (128.center) to (133.center);
		\draw (128.center) to (129.center);
		\draw (129.center) to (135.center);
		\draw (133.center) to (135.center);
		\draw (134.center) to (136.center);
		\draw (139.center) to (140.center);
		\draw (140.center) to (142.center);
		\draw (142.center) to (141.center);
		\draw (141.center) to (139.center);
		\draw (144.center) to (149.center);
		\draw (144.center) to (145.center);
		\draw (145.center) to (151.center);
		\draw (149.center) to (151.center);
		\draw (150.center) to (152.center);
		\draw (157.center) to (159.center);
		\draw (159.center) to (158.center);
		\draw (164.center) to (158.center);
		\draw (164.center) to (157.center);
		\draw (160.center) to (161.center);
		\draw (36.center) to (117.center);
		\draw (124.center) to (155.center);
		\draw (37.center) to (39.center);
		\draw (46.center) to (49.center);
		\draw (52.center) to (61.center);
		\draw (68.center) to (85.center);
		\draw (92.center) to (95.center);
		\draw (98.center) to (107.center);
		\draw (114.center) to (118.center);
		\draw (125.center) to (127.center);
		\draw (130.center) to (143.center);
		\draw (146.center) to (156.center);
		\draw (28.center) to (40.center);
		\draw (47.center) to (59.center);
		\draw (60.center) to (62.center);
		\draw (69.center) to (86.center);
		\draw (93.center) to (105.center);
		\draw (106.center) to (108.center);
		\draw (115.center) to (137.center);
		\draw (138.center) to (153.center);
		\draw (162.center) to (165.center);
		\draw (163.center) to (166.center);
		\draw (167.center) to (154.center);
		\draw [style=red] (12.center) to (21.center);
		\draw [style=red] (26.center) to (74.center);
		\draw [style=red] (75.center) to (171.center);
		\draw (172.center) to (91.center);
		\draw [style=red] (91.center) to (99.center);
		\draw [style=red] (100.center) to (113.center);
		\draw (113.center) to (123.center);
		\draw [style=red] (123.center) to (131.center);
		\draw [style=red] (132.center) to (147.center);
		\draw [style=red] (148.center) to (161.center);
		\draw (161.center) to (169.center);
		\draw (152.center) to (170.center);
		\draw [style=red] (136.center) to (173.center);
		\draw [style=red] (174.center) to (152.center);
		\draw (5.center) to (58.center);
		\draw [style=red] (58.center) to (76.center);
		\draw [style=red] (77.center) to (104.center);
		\draw (104.center) to (136.center);
		\draw [style=red] (22.center) to (35.center);
		\draw [style=red] (45.center) to (53.center);
		\draw [style=red] (54.center) to (67.center);
		\draw (67.center) to (71.center);
		\draw [style=red] (71.center) to (175.center);
		\draw [style=red] (176.center) to (177.center);
		\draw [style=red] (178.center) to (179.center);
		\draw [style=red] (180.center) to (168.center);
		\draw [style=red dtt] (177.center) to (178.center);
		\draw [style=red dtt] (179.center) to (180.center);
		\draw [style=dtt] (154.center) to (153.center);
		\draw [style=dtt] (138.center) to (137.center);
		\draw [style=dtt] (105.center) to (106.center);
		\draw [style=red dtt] (175.center) to (176.center);
		\draw [style=dtt] (59.center) to (60.center);
		\draw [style=dtt] (27.center) to (28.center);
		\draw [style=doubleline] (181.center) to (182.center);
		\draw [style=doubleline] (182.center) to (183.center);
		\draw [style=doubleline] (183.center) to (184.center);

    \node [style=none] (0) at (0, 0) {};
		\node [style=none] (1) at (0, -1) {};
		\node [style=none] (2) at (0, -2) {};
		\node [style=none] (3) at (0, -3) {};
		\node [style=none] (4) at (0, -4) {};
		\node [style=none] (5) at (0, -5) {};
		\node [style=none] (6) at (0.75, 0) {};
		\node [style=none] (7) at (0.75, -1) {};
		\node [style=none] (8) at (0.75, 0.25) {};
		\node [style=none] (9) at (1.25, -1.25) {};
		\node [style=none] (10) at (0.75, -1.25) {};
		\node [style=none] (11) at (1, -1.25) {};
		\node [style=none] (12) at (1, -3) {};
		\node [style=none] (13) at (1.25, 0) {};
		\node [style=none] (14) at (1.25, -1) {};
		\node [style=none] (15) at (1.25, 0.25) {};
		\node [style=none] (17) at (1.75, -1) {};
		\node [style=none] (18) at (1.75, -0.75) {};
		\node [style=none] (19) at (2.25, -0.75) {};
		\node [style=none] (20) at (2.25, -1) {};
		\node [style=none] (21) at (1.75, -3) {};
		\node [style=none] (22) at (2.25, -3) {};
		\node [style=none] (23) at (1.75, -3.25) {};
		\node [style=none] (24) at (2, -3.25) {};
		\node [style=none] (25) at (2.25, -3.25) {};
		\node [style=none] (26) at (2, -4) {};
		\node [style=none] (27) at (1.75, -2) {};
		\node [style=none] (28) at (2.25, -2) {};
		\node [style=none] (29) at (2.75, 0) {};
		\node [style=none] (30) at (2.75, -1) {};
		\node [style=none] (31) at (2.75, 0.25) {};
		\node [style=none] (32) at (3.25, -1.25) {};
		\node [style=none] (33) at (2.75, -1.25) {};
		\node [style=none] (34) at (3, -1.25) {};
		\node [style=none] (35) at (3, -3) {};
		\node [style=none] (36) at (3.25, 0) {};
		\node [style=none] (37) at (3.25, -1) {};
		\node [style=none] (38) at (3.25, 0.25) {};
		\node [style=none] (39) at (3.75, -1) {};
		\node [style=none] (40) at (3.75, -2) {};
		\node [style=none] (41) at (3.75, -0.75) {};
		\node [style=none] (42) at (4.25, -2.25) {};
		\node [style=none] (43) at (3.75, -2.25) {};
		\node [style=none] (44) at (4, -2.25) {};
		\node [style=none] (45) at (4, -3) {};
		\node [style=none] (46) at (4.25, -1) {};
		\node [style=none] (47) at (4.25, -2) {};
		\node [style=none] (48) at (4.25, -0.75) {};
		\node [style=none] (49) at (4.75, -1) {};
		\node [style=none] (50) at (4.75, -0.75) {};
		\node [style=none] (51) at (5.25, -0.75) {};
		\node [style=none] (52) at (5.25, -1) {};
		\node [style=none] (53) at (4.75, -3) {};
		\node [style=none] (54) at (5.25, -3) {};
		\node [style=none] (55) at (4.75, -3.25) {};
		\node [style=none] (56) at (5, -3.25) {};
		\node [style=none] (57) at (5.25, -3.25) {};
		\node [style=none] (58) at (5, -5) {};
		\node [style=none] (59) at (4.75, -2) {};
		\node [style=none] (60) at (5.25, -2) {};
		\node [style=none] (61) at (5.75, -1) {};
		\node [style=none] (62) at (5.75, -2) {};
		\node [style=none] (63) at (5.75, -0.75) {};
		\node [style=none] (64) at (6.25, -2.25) {};
		\node [style=none] (65) at (5.75, -2.25) {};
		\node [style=none] (66) at (6, -2.25) {};
		\node [style=none] (67) at (6, -3) {};
		\node [style=none] (68) at (6.25, -1) {};
		\node [style=none] (69) at (6.25, -2) {};
		\node [style=none] (70) at (6.25, -0.75) {};
		\node [style=none] (71) at (7, -3) {};
		\node [style=none] (72) at (6.75, -5.25) {};
		\node [style=none] (73) at (7.25, -5.25) {};
		\node [style=none] (74) at (6.75, -4) {};
		\node [style=none] (75) at (7.25, -4) {};
		\node [style=none] (76) at (6.75, -5) {};
		\node [style=none] (77) at (7.25, -5) {};
		\node [style=none] (78) at (6.75, -3.75) {};
		\node [style=none] (79) at (7.25, -3.75) {};
		\node [style=none] (80) at (7, -3.75) {};
		\node [style=none] (81) at (7.75, -3.75) {};
		\node [style=none] (82) at (8.25, -3.75) {};
		\node [style=none] (83) at (8.25, -4.25) {};
		\node [style=none] (84) at (7.75, -4.25) {};
		\node [style=none] (85) at (8.75, -1) {};
		\node [style=none] (86) at (8.75, -2) {};
		\node [style=none] (87) at (8.75, -0.75) {};
		\node [style=none] (88) at (9.25, -2.25) {};
		\node [style=none] (89) at (8.75, -2.25) {};
		\node [style=none] (90) at (9, -2.25) {};
		\node [style=none] (91) at (9, -4) {};
		\node [style=none] (92) at (9.25, -1) {};
		\node [style=none] (93) at (9.25, -2) {};
		\node [style=none] (94) at (9.25, -0.75) {};
		\node [style=none] (95) at (9.75, -1) {};
		\node [style=none] (96) at (9.75, -0.75) {};
		\node [style=none] (97) at (10.25, -0.75) {};
		\node [style=none] (98) at (10.25, -1) {};
		\node [style=none] (99) at (9.75, -4) {};
		\node [style=none] (100) at (10.25, -4) {};
		\node [style=none] (101) at (9.75, -4.25) {};
		\node [style=none] (102) at (10, -4.25) {};
		\node [style=none] (103) at (10.25, -4.25) {};
		\node [style=none] (104) at (10, -5) {};
		\node [style=none] (105) at (9.75, -2) {};
		\node [style=none] (106) at (10.25, -2) {};
		\node [style=none] (107) at (10.75, -1) {};
		\node [style=none] (108) at (10.75, -2) {};
		\node [style=none] (109) at (10.75, -0.75) {};
		\node [style=none] (110) at (11.25, -2.25) {};
		\node [style=none] (111) at (10.75, -2.25) {};
		\node [style=none] (112) at (11, -2.25) {};
		\node [style=none] (113) at (11, -4) {};
		\node [style=none] (114) at (11.25, -1) {};
		\node [style=none] (115) at (11.25, -2) {};
		\node [style=none] (116) at (11.25, -0.75) {};
		\node [style=none] (117) at (11.75, 0) {};
		\node [style=none] (118) at (11.75, -1) {};
		\node [style=none] (119) at (11.75, 0.25) {};
		\node [style=none] (120) at (12.25, -1.25) {};
		\node [style=none] (121) at (11.75, -1.25) {};
		\node [style=none] (122) at (12, -1.25) {};
		\node [style=none] (123) at (12, -4) {};
		\node [style=none] (124) at (12.25, 0) {};
		\node [style=none] (125) at (12.25, -1) {};
		\node [style=none] (126) at (12.25, 0.25) {};
		\node [style=none] (127) at (12.75, -1) {};
		\node [style=none] (128) at (12.75, -0.75) {};
		\node [style=none] (129) at (13.25, -0.75) {};
		\node [style=none] (130) at (13.25, -1) {};
		\node [style=none] (131) at (12.75, -4) {};
		\node [style=none] (132) at (13.25, -4) {};
		\node [style=none] (133) at (12.75, -4.25) {};
		\node [style=none] (134) at (13, -4.25) {};
		\node [style=none] (135) at (13.25, -4.25) {};
		\node [style=none] (136) at (13, -5) {};
		\node [style=none] (137) at (12.75, -2) {};
		\node [style=none] (138) at (13.25, -2) {};
		\node [style=none] (139) at (13.75, -4.75) {};
		\node [style=none] (140) at (14.25, -4.75) {};
		\node [style=none] (141) at (13.75, -5.25) {};
		\node [style=none] (142) at (14.25, -5.25) {};
		\node [style=none] (143) at (14.75, -1) {};
		\node [style=none] (144) at (14.75, -0.75) {};
		\node [style=none] (145) at (15.25, -0.75) {};
		\node [style=none] (146) at (15.25, -1) {};
		\node [style=none] (147) at (14.75, -4) {};
		\node [style=none] (148) at (15.25, -4) {};
		\node [style=none] (149) at (14.75, -4.25) {};
		\node [style=none] (150) at (15, -4.25) {};
		\node [style=none] (151) at (15.25, -4.25) {};
		\node [style=none] (152) at (15, -5) {};
		\node [style=none] (153) at (14.75, -2) {};
		\node [style=none] (154) at (15.25, -2) {};
		\node [style=none] (155) at (15.75, 0) {};
		\node [style=none] (156) at (15.75, -1) {};
		\node [style=none] (157) at (15.75, 0.25) {};
		\node [style=none] (158) at (16.25, -1.25) {};
		\node [style=none] (159) at (15.75, -1.25) {};
		\node [style=none] (160) at (16, -1.25) {};
		\node [style=none] (161) at (16, -4) {};
		\node [style=none] (162) at (16.25, 0) {};
		\node [style=none] (163) at (16.25, -1) {};
		\node [style=none] (164) at (16.25, 0.25) {};
		\node [style=none] (165) at (17, 0) {};
		\node [style=none] (166) at (17, -1) {};
		\node [style=none] (167) at (17, -2) {};
		\node [style=none] (168) at (17, -3) {};
		\node [style=none] (169) at (17, -4) {};
		\node [style=none] (170) at (17, -5) {};
		\node [style=none] (171) at (7.75, -4) {};
		\node [style=none] (172) at (8.25, -4) {};
		\node [style=none] (173) at (13.75, -5) {};
		\node [style=none] (174) at (14.25, -5) {};
		\node [style=none] (175) at (9.75, -3) {};
		\node [style=none] (176) at (10.25, -3) {};
		\node [style=none] (177) at (12.75, -3) {};
		\node [style=none] (178) at (13.25, -3) {};
		\node [style=none] (179) at (14.75, -3) {};
		\node [style=none] (180) at (15.25, -3) {};
		\node [style=none] (181) at (8, -4.25) {};
		\node [style=none] (182) at (8, -6) {};
		\node [style=none] (183) at (14, -6) {};
		\node [style=none] (184) at (14, -5.25) {};
		\node [style=new style 0] (185) at (-0.5, 0) {$x_1$};
		\node [style=new style 0] (186) at (-0.5, -1) {$x_2$};
		\node [style=new style 0] (187) at (-0.5, -2) {$x_3$};
		\node [style=new style 0] (188) at (-0.5, -3) {$|0\rangle$};
		\node [style=new style 0] (189) at (-0.5, -4) {$|0\rangle$};
		\node [style=new style 0] (190) at (-0.5, -5) {$|0\rangle$};
		\node [style=new style 0] (191) at (17.5, 0) {$x_1$};
		\node [style=new style 0] (192) at (17.5, -1) {$x_2$};
		\node [style=new style 0] (193) at (17.5, -2) {$x_3$};
		\node [style=new style 0] (194) at (17.5, -3) {$E$};
		\node [style=new style 0] (195) at (17.5, -4) {$|0\rangle$};
		\node [style=new style 0] (196) at (17.5, -5) {$|0\rangle$};
		\node [style=none] (197) at (1, -0.5) {A};
		\node [style=none] (198) at (2, -2) {C};
		\node [style=none] (199) at (3, -0.5) {A};
		\node [style=none] (200) at (4, -1.5) {B};
		\node [style=none] (201) at (5, -2) {D};
		\node [style=none] (202) at (6, -1.5) {B};
		\node [style=none] (203) at (7, -4.5) {E};
		\node [style=none] (204) at (8, -4) {$G$};
		\node [style=none] (205) at (9, -1.5) {B};
		\node [style=none] (206) at (10, -2.5) {D};
		\node [style=none] (207) at (11, -1.5) {B};
		\node [style=none] (208) at (12, -0.5) {A};
		\node [style=none] (209) at (13, -2.5) {C};
		\node [style=none] (210) at (14, -5) {$Z$};
		\node [style=none] (211) at (15, -2.5) {C};
		\node [style=none] (212) at (16, -0.5) {A};
        \node [style=plus] (213) at (1, -3) {};
		\node [style=plus] (214) at (2, -4) {};
		\node [style=plus] (215) at (3, -3) {};
		\node [style=plus] (216) at (4, -3) {};
		\node [style=plus] (217) at (5, -5) {};
		\node [style=plus] (218) at (6, -3) {};
		\node [style=plus] (219) at (7, -3) {};
		\node [style=plus] (220) at (9, -4) {};
		\node [style=plus] (221) at (10, -5) {};
		\node [style=plus] (222) at (11, -4) {};
		\node [style=plus] (223) at (12, -4) {};
		\node [style=plus] (224) at (13, -5) {};
		\node [style=plus] (225) at (15, -5) {};
		\node [style=plus] (226) at (16, -4) {};
  \end{tikzpicture}

%% file: PSPACE-completeness/newDAGbinary_G_.tex
\begin{tikzpicture}[main/.style = {draw, circle,minimum size=1.6cm}, scale=1.35] 
\hspace*{2pt}

\node[main] (q) at (0,0) {\tiny$q=1,1$}; 

\node[main] (11) at (-1,1) {\tiny1,2}; 
\node[main] (12) at (1,1) {\tiny2,1}; 

\node[main] (21) at (-2,2) {\tiny1,3}; 
\node[main] (22) at (0,2) {\tiny2,2}; 
\node[main] (23) at (2,2) {\tiny3,1}; 

\node[main] (31) at (-4,4) {\tiny1,$P$-1}; 
\node[main] (32) at (-2,4) {\tiny2,$P$-2}; 
\node[main] (34) at (4,4) {\tiny$P$-1,1}; 

\node[main] (41) at (-5,5) {\tiny1,$P$}; 
\node[main] (42) at (-3,5) {\tiny2,$P$-1}; 
\node[main] (44) at (5,5) {\tiny$P$,1}; 

\node[main] (51) at (-4,6) {\tiny2,$P$}; 
\node[main] (52) at (-2,6) {\tiny3,$P$-1}; 
\node[main] (54) at (4,6) {\tiny$P$,2};

\node[main] (61) at (-2,8) {\tiny$P$-2,$P$}; 
\node[main] (62) at (0,8) {\tiny$P$-1,$P$-1}; 
\node[main] (63) at (2,8) {\tiny$P$,$P$-2}; 

\node[main] (71) at (-1,9) {\tiny$P$-1,$P$}; 
\node[main] (72) at (1,9) {\tiny$P$,$P$-1}; 

\node[main] (q') at (0,10) {\tiny$q'=P,P$};

\draw[->] (q) -- (11);
\draw[->] (q) -- (12);
\draw[->] (11) -- (21);
\draw[->] (11) -- (22);
\draw[->] (12) -- (22);
\draw[->] (12) -- (23);

\draw[dashed] (21) -- (31);
\draw[dashed] (23) -- (34);

\draw[dotted] (32) -- (34);
\draw[dotted] (42) -- (44);
\draw[dotted] (52) -- (54);

\draw[->] (31) -- (41);
\draw[->] (31) -- (42);
\draw[->] (32) -- (42);
\draw[->] (34) -- (44);

\draw[->] (41) -- (51);
\draw[->] (42) -- (51);
\draw[->] (42) -- (52);
\draw[->] (44) -- (54);

\draw[dashed] (51) -- (61);
\draw[dashed] (54) -- (63);

\draw[->] (71) -- (q');
\draw[->] (72) -- (q');
\draw[->] (61) -- (71);
\draw[->] (62) -- (71);
\draw[->] (62) -- (72);
\draw[->] (63) -- (72);

\end{tikzpicture} 
\vspace{-2em}

%% file: PSPACE-completeness/newDAG_G_.tex
\begin{tikzpicture}[main/.style = {draw, circle, minimum size = 7mm}] 

\node[] (G) at (0,-.8) {$G$}; 

\node[main] (q) at (0,0) {$q$}; 

\node[] (Gdiam) at (0,1) {$G^\diamond_P$};

\node[main] (q') at (0,2) {$q'$};

\draw[dashed] (0.5,-1) -- (q);
\draw[dashed] (-0.5,-1) -- (q);

\draw[dashed] (q) -- (.8,1) -- (q');
\draw[dashed] (q) -- (-.8,1) -- (q');

\end{tikzpicture} 